\documentclass[11pt]{article}
\usepackage[T1]{fontenc}
\usepackage{booktabs}
\usepackage{nicefrac}       
\usepackage{microtype}      
\usepackage{amsmath,amsbsy,amsfonts,amssymb,amsthm,units}
\usepackage{graphicx}
\usepackage{tablefootnote}
\usepackage{color,cases}
\usepackage{tikz}
\usetikzlibrary{calc,shapes,intersections}
\usepackage{wrapfig}
\usepackage{subfigure}
\usepackage[margin=1.1in]{geometry}
\usepackage{xcolor}
\definecolor{DarkGreen}{rgb}{0.1,0.5,0.1}
\definecolor{DarkRed}{rgb}{0.5,0.1,0.1}
\definecolor{DarkBlue}{rgb}{0.1,0.1,0.5}
\definecolor{Black}{rgb}{0.0,0.0,0.0}
\usepackage{hyperref}
\hypersetup{
    colorlinks=true,       
    linkcolor=DarkBlue,          
    citecolor=DarkBlue,        
    filecolor=DarkBlue,      
    urlcolor=DarkBlue,          
    pdftitle={Fast binary embeddings, and quantized compressed sensing with structured matrices},
    pdfauthor={Thang Huynh, Rayan Saab},
}
\usepackage{nicefrac}
\usepackage{algorithm}
\usepackage{algorithmic}
\usepackage{xspace}
\usepackage{wrapfig}
\usepackage{multicol}

\usepackage{caption}
\newtheorem{theorem}{Theorem}[section]

\newtheorem{lemma}[theorem]{Lemma}

\newtheorem{proposition}[theorem]{Proposition}
\theoremstyle{definition}

\newtheorem{definition}[theorem]{Definition}
\newtheorem{construction}[theorem]{Construction}
\newtheorem{example}[theorem]{Example}
\newtheorem{remark}{Remark}
\numberwithin{equation}{section}

\newcommand{\<}{\langle}
\renewcommand{\>}{\rangle}
\newcommand{\EE}{\mathbb{E}}
\newcommand{\PP}{\mathbb{P}}
\newcommand{\R}{\mathbb{R}}
\newcommand{\C}{\mathbb{C}}
\newcommand{\Z}{\mathbb{Z}}
\newcommand{\A}{\mathcal{A}}
\newcommand{\Q}{\mathcal{Q}}
\newcommand{\T}{\mathcal{T}}
\newcommand{\X}{\mathcal{X}}

\newcommand\inner[1]{\langle #1 \rangle}
\newcommand\sign{\operatorname{sign}}
\newcommand\rad{\operatorname{rad}}

\newcommand{\sparseset}{\Sigma_k^n}

\title{Fast binary embeddings, and quantized compressed sensing with structured matrices}
\author{Thang Huynh \thanks{Department of Mathematics, University of California San Diego, e-mail: tlh007@ucsd.edu} \and Rayan Saab \thanks{Department of Mathematics, University of California San Diego, e-mail: rsaab@ucsd.edu.}}
\date{}

\begin{document}
\maketitle

\begin{abstract}
This paper deals with two related problems, namely distance-preserving binary embeddings and quantization for compressed sensing. First, we propose fast methods to replace points from a subset $\mathcal{X} \subset \R^n$, associated with the Euclidean metric, with points in the cube $\{\pm 1\}^m$ and we associate the cube with a pseudo-metric that approximates Euclidean distance among points in $\mathcal{X}$. Our methods rely on quantizing fast Johnson-Lindenstrauss embeddings based on \emph{bounded orthonormal systems} and \emph{partial circulant ensembles}, both of which admit fast transforms. 
Our quantization methods utilize \emph{noise-shaping}, and include Sigma-Delta schemes and distributed noise-shaping schemes. The resulting approximation errors  decay polynomially and exponentially fast in $m$, depending on the embedding method. This dramatically outperforms the current decay rates associated with binary embeddings and Hamming distances. Additionally, it is the first such binary embedding result that applies to fast Johnson-Lindenstrauss maps while preserving $\ell_2$ norms. 

Second, we again consider  noise-shaping schemes, albeit this time to quantize compressed sensing measurements arising from bounded orthonormal ensembles and partial circulant matrices. We show that these methods yield a reconstruction error that again decays with the number of measurements (and bits), when using convex optimization for reconstruction. Specifically, for Sigma-Delta schemes, the error decays polynomially in the number of measurements, and it decays exponentially  for distributed noise-shaping schemes based on beta encoding. These results are near optimal and the first of their kind dealing with bounded orthonormal systems. 
\end{abstract}

%

\section{Introduction}
In signal processing and machine learning, given a linear map $y = A x$, where $x$ is an unknown signal in $\C^n$ and $A$ is an $m\times n$ matrix in $\C^{m\times n}$, one would often like to quantize or digitize $y$, i.e., map each entry $y_i$ of $y$ to a finite discrete set, say $\mathcal{A}$. 

On the one hand, viewing $A$ as a measurement operator, quantization is the step that allows us to move from the analog world of continuous measurements to the digital world of bit-streams. It thereby allows computer processing of the signals, as well as their storage and transmission. On the other hand, one can view $A$ as a linear embedding, taking some vectors $x\in \mathcal{T} \subset \R^n$ to their image in $\R^m$. For a vector $x\in \mathcal{T}$,  the \emph{quantized} version of $Ax$ can be viewed as a \emph{non-linear embedding} of $x\in \R^n$ into $\mathcal{A}^m$. It can therefore be used to expedite such tasks as information retrieval (say via nearest neighbor searches) provided the action of $A$ and the quantization both admit fast computation, and provided "distance" computations in the embedding space can be done efficiently as well. 

\subsection{A brief introduction to compressed sensing}

While in the remainder of the paper we address binary embeddings before we address quantization  in the compressed sensing context, we find it more convenient here to first introduce  compressed sensing. Suppose $x \in \C^n$ is an unknown vector  which we wish to reconstruct from $m$ linear measurements. Specifically, suppose we collect the measurements \[y_j = \inner{a_j, x} + w_j\] where each $a_j$, $j=1,...,m$, is a known measurement vector in $\C^n$ and $w_j$ denotes unknown noise. When $x$ is sparse (i.e., most of its coefficients are zero) or compressible (i.e., well approximated by sparse vectors) $x$ can be recovered (or approximated) from the measurements, even when $m\ll n$. Indeed, the study of recovery algorithms and reconstruction guarantees from such measurements is the purview of  \emph{compressed sensing} (CS) and one of the most popular reconstruction methods in this setting is $\ell_1$-minimization. Here, the vector $x$ is approximated by $\hat{x}$, the solution to 
\begin{equation}\label{l1-min}
	\min_{z\in\C^n} \|z\|_1 \quad\mbox{subject to}\quad \|A z - y\|_2 \leq \eta.
\end{equation}
Above $A$ is the $m\times n$ matrix with $a_j$ as its $j$th row, $y = (y_j)_{j=1}^m$ is the vector of acquired measurements, and $\eta$ is an upper bound on the norm of the noise vector $w = (w_j)_{j=1}^m$.

Some of the earliest results on compressed sensing (e.g., \cite{CTR2006, Donoho2006, MPT2008, BDD2008}) show that for a certain class of \emph{subgaussian} random measurement matrices $A$, including those with i.i.d. standard Gaussian or $\pm 1$ Bernoulli entries, with high probability on the draw of the matrix and uniformly on $x$, the solution $\widehat{x}$ to \eqref{l1-min} obeys
\begin{equation}\label{error bound}
	\|x - \widehat{x}\|_2 \lesssim \eta+ \frac{\sigma_k(x)_1}{\sqrt{k}},
\end{equation}
provided $m \gtrsim k\log(n/k).$ Such guarantees are often based  on the measurement matrix $A$ satisfying the so-called Restricted Isometry Property (RIP), defined precisely in Section \ref{sec:math_prelim}. 

%
The above results on subgaussian random matrices are mathematically important both for the techniques they introduce and as a proof of concept that compressed sensing is viable. Nevertheless, for physical reasons one may not have full design control over the measurement matrix in practice. Moreover, for practical reasons (including reconstruction speed), one may require that the measurement matrices admit fast multiplications. Such reasons motivated the study of structured random matrices in the compressed sensing context (see, e.g., \cite{CT2006, TWDBB2006, RV2008, Rauhut2008, FR2013}). Perhaps two of the most popular classes of structured random measurement matrices are those drawn from a \emph{bounded orthogonal ensemble} (BOE) or from a \emph{partial circulant ensemble} (PCE).  Indeed, studying the interplay between  structured random matrices and quantization is  one focus of this paper.

Herein, we are interested in the case when $A$ is drawn from a bounded orthogonal ensemble, or from a partial circulant ensemble, both of which we now introduce.

\begin{definition}[Bounded Orthogonal Ensemble (BOE)]\label{def:BOE}
Let $\frac{1}{\sqrt{n}}U \in \C^{n\times n}$ be any unitary matrix with $|U_{ij}| \leq 1$ for all entries $U_{ij}$ of $U$. A matrix $A \in \C^{m\times n}$ is drawn from the bounded orthogonal ensemble associated with $U$ by picking each row $a_i$ of the matrix $A$ uniformly and independently from the set of all rows of $U$. 
\end{definition}
Examples of $U$ include the $n\times n$ discrete Fourier matrix and the $n\times n$ Hadamard matrix.
For physical reasons, BOEs arise naturally in various important applications of compressed sensing, including those involving Magnetic Resonance Imaging (MRI) (see, e.g., \cite{HHL2011, LDP2007, Murphy2012, Vasanawala2010}). An additional practical benefit of BOEs, e.g., those based on the Fourier or Hadamard transforms, is implementation speed.  When viewed as linear operators, such matrices admit fast implementations with a number of additions and multiplications that scales like $n \log{n}$. This is in contrast to the cost of standard matrix vector multiplication which scales like $n^2$. As such, BOEs have also appeared in various results involving fast Johnson-Lindenstrauss embeddings (e.g., \cite{AL2013, AC2009, KW2011}).

\begin{definition}[Partial Circulant Ensemble (PCE)]\label{def:PCE}
For $z\in\C^n,$ the circulant matrix $H_z \in\C^{n\times n}$ is given by its action $H_z x = z*x$, where $*$ denotes convolution. Fix $\Omega \subset\{1, 2, \ldots, n\}$ of size $m$ arbitrarily. A matrix $A$ is drawn from the partial circulant ensemble associated with $\Omega$  by choosing a vector $\sigma$ uniformly at random on $\{-1, 1\}^n$ (i.e., the entries $\sigma_i$ of $\sigma$ are independent $\pm 1$ Bernoulli random variables), and setting the rows of  $A$ to be the rows of $H_\sigma$ indexed by $\Omega$. 
\end{definition}
Partial Circulant Ensembles also appear naturally in various applications of compressed sensing, including those involving radar and wireless channel estimation mainly due to the fact that the action of a circulant matrix on a vector is equivalent to its convolution with a row of the circulant matrix. We refer to \cite{HBRN2010, RRT2012,Romberg2009, FKS2017} for a discussion on these applications. As with BOEs, PCEs also admit fast transformations,  as convolution is essentially diagonalized by the Fourier transform, which admits fast implementation. As such, PCEs also admit an implementation with a number of addition and multiplications that scales with $n\log n$ as opposed to $n^2$. For this reason PCEs, like BOEs, feature prominently in constructions for fast Johnson-Lindenstrauss embeddings. An additional benefit in applications is that the memory required for storing a PCE scales like $n$ (versus $mn$ for an unstructured matrix).

Cand\'es and Tao were the first to study structured random matrices  drawn from a BOE, in the compressed sensing context. In \cite{CT2006}, they show that such matrices satisfy the restricted isometry property with high probability, provided that the number of measurements $m \gtrsim k\log^6(n).$ Later, many researchers (e.g., \cite{RV2008, CGV2013, B2014, HR2016})  improved the dependence of the number of measurements $m$ on $n$, i.e., they improved the logarithmic factors. To our knowledge, the current best result is by Haviv and Regev \cite{HR2016} and achieves a lower bound $m \gtrsim k\log^2(n)$. Similarly, many researchers have studied recovery guarantees for compressed sensing with random measurement matrices  drawn from the partial circulant ensembles (e.g., \cite{HBRN2010, PRT2013, KMR2014, MRW2016}), eventually showing a similar dependence of $m$ on $k$ and $n$, i.e., $m\gtrsim k\log^4{n}$ in \cite{KMR2014}.  Our own results will apply to both BOEs and PCEs (in both the compressed sensing and binary embedding contexts) with the added caveat of randomizing the signs of their rows. It is therefore useful to introduce the following construction. 
\begin{construction}\label{distribution}
An admissible distribution on $m\times n$ matrices corresponds to constructing  $\Phi \in \C^{m\times n}$ as follows.
\begin{enumerate}
\item Draw $A$ either from a bounded orthogonal ensemble associated with a matrix $U$ as in Definition \ref{def:BOE} or from a partial circulant ensemble associated with a set $\Omega$ of size $m$ as in Definition \ref{def:PCE}. Let $a_j$ be the $j$-th row of $A$.
\item Let $\epsilon_j$ be independent $\pm 1$ Bernoulli random variables which are also independent of $A$.
\item Let $\Phi$ be an $m\times n$ matrix whose $j$-th row is $\epsilon_j a_j$.
\end{enumerate}
\end{construction}

\subsection{A brief introduction to quantization}
Consider the linear measurements $y = A x$ of a signal $x$. Quantization is the process by which the measurement vector $y$ is replaced by a vector of elements from a finite set, say $\mathcal{A}$, known as the quantization alphabet. The finiteness of $\mathcal{A}$ allows its elements to be represented with finite binary strings, which in turn allows  digital storage and processing. Indeed, such digital processing is 
necessary in the field of binary embedding \cite{JLBB2013, PV2013, PV2014, YCP2015} and compressed sensing (see, e.g., \cite{CGKRO2015}), where careful selection of $A$ and $\mathcal{A}$ can yield faster, memory-efficient algorithms. In binary embedding one wishes to design a quantization map $\Q:\C^m \to \A^m$ which approximately preserves the distance between two signals (see below). On the other hand, in compressed sensing one typically requires a reconstruction algorithm $\mathcal{R}: \mathcal{A}^m \to \C^n$ such that given the quantized measurements $\Q(A x),$  $\hat{x} = \mathcal{R}(\Q(A x))$ ensures a small reconstruction error $\|x - \hat{x}\|_2.$
Various choices of quantization maps and reconstruction algorithms have been proposed in the context of binary embedding and compressed sensing. These have ranged from the most intuitive quantization approaches, namely memoryless scalar quantization \cite{JLBB2013, PV2013, DJR2017}, to more sophisticated approaches based on  noise shaping quantization, including $\Sigma\Delta$  quantization \cite{GLPSY2013, KSW2012, KSY2014} and distributed noise shaping quantization \cite{Chou2013,H2016}, as well as others \cite{BFNPW2017, JC2013}. Indeed, the afore-mentioned noise shaping quantization techniques combine computational simplicity and the ability to yield favorable error guarantees as a function of the number of measurements. We will use them to obtain our  results and thus provide a necessary, more detailed, overview in Section \ref{sec:quantization}.

\subsection{A brief introduction to binary embeddings}

Low distortion embeddings have played an important role recently in  signal processing and machine learning  as they transform high dimensional signals into low-dimensional ones while preserving geometric properties. The benefits of such embeddings are direct consequences of the dimensionality reduction and they include the potential for reduced storage space and computational time associated with the (embedded) signals. Perhaps one of the most important embeddings is given by the Johnson-Lindenstrauss (JL) lemma \cite{JL} which shows that one can embed $N$ points in $\R^n$ into a $m=O(\delta^{-2}\log N)$-dimensional space and simultaneously preserve pairwise Euclidean distance up to $\delta$-Lipschitz distortion. 
More recently, binary embeddings   have also attracted a lot of attention in the signal processing and machine learning communities (e.g., \cite{WTF2009, RL2009, SH2009, LWKC2011, GL2013, YBKGC2015}), in part due to theoretical interest, and in part motivated by the potential benefits of further reductions in memory requirements and computational time. Roughly speaking, they are nonlinear embeddings that map high dimensional signals into a discrete cube in lower-dimensional space. As each signal is now represented by a short binary code, the memory needed for storing the entire data set is reduced considerably, as is the cost of computation.

To be more precise, let $\mathcal{T}$ be a set of vectors in $\R^n$ and let $\{-1, 1\}^m$ be the binary cube. Given a distance $d_{\T}$ on $\T$, an $\alpha$-binary embedding of $\T$ is a map $f: \T \to \{-1, 1\}^m$ with an associated function $d$ on $\{-1, 1\}^m$ so that 
\[
	\left|d(f(x), f(\tilde{x})) - d_{\T}(x, \tilde{x})\right| \leq \alpha, \qquad \text{for all } x, \tilde{x} \in \T.
\]
Recent works have shown that there exist such $\alpha$-binary embeddings \cite{JLBB2013, PV2013, PV2014, YCP2015}. For $\T$ a finite set of $N$ points in the unit sphere $S^{n-1}$, endowed with the angular distance $d_{\T}(x, \tilde{x}) = \arccos(\|x\|_2^{-1}\|\tilde{x}\|_2^{-1}\inner{x, \tilde{x}})$, these works consider the  map
\begin{equation}\label{sign function}
	f_A(x) = \sign(Ax),
\end{equation}
where $A$ is an $m\times n$ matrix whose entries are independent standard Gaussian random variables. Then with probability exceeding $1 - \eta$, $f_A$ is an $\alpha$-binary embedding for $\T$ into $\{-1, 1\}^m$ associated with the normalized Hamming distance $d_H(q, \tilde{q}) = \frac{1}{2m}\sum_{i}|q_i - \tilde{q}_i|$, provided that ${m\gtrsim \alpha^{-2}\log(N/\eta).}$ 

While achieving the optimal bit complexity for \emph{Hamming distances} as shown in \cite{YCP2015}, the above binary embeddings rely on $A$ being Gaussian. In general, they suffer from the following drawbacks.
\begin{itemize}
\item[] {\bf Drawback (i):} To execute an embedding, one must incur a memory cost  of order $mn$ to store $A$, and must also incur the full computational cost of dense matrix-vector multiplication each time a point is embedded. 
\item[] {\bf Drawback (ii):} The embedding completely loses all magnitude information associated with the original points. That is, all points $c x \in \R^n$ with $c > 0$ embed to the same image. 
\end{itemize}
To address the first point above, researchers have tried to design other binary embeddings which can be implemented more efficiently. For example, such embeddings include circulant binary embedding \cite{YBKGC2015, OTH2017, DS2016}, structured hashed projections \cite{CCBJKL2016}, binary embeddings based on Walsh-Hadamard matrices and partial Gaussian Toeplitz matrices with random column flips \cite{YCP2015}. To our knowledge, all of these results do not address the second point above; they use the sign function \eqref{sign function} (which is an instance of Memoryless Scalar Quantization) to \emph{quantize} the measurements. While  simple, this quantization method cannot yield much better distance preservation when we have more measurements, i.e., when $m$ increases, as discussed in {Section \ref{sec:quantization}}. 
\begin{itemize}
\item[] {\bf Drawback (iii):} At best, the approximation error associated with the embedding distance decays slowly as  $m$, the embedding dimension (and number of bits) increases.
\end{itemize}
We resolve these issues in this paper. 


\subsection{Notation}
A vector is $k$-sparse if it has at most $k$ nonzero entries, i.e., it belongs to the set $${\sparseset := \{ z\in \C^n, |\mbox{supp}(z)| \leq k\}.}$$ Roughly speaking, a compressible vector $x$ is one that is close to a $k$-sparse vector. More precisely, we can use the best $k$-term approximation error of $x$ in $\ell_1$, denoted by $\sigma_k(x)_1 := \min_{v\in\sparseset} \|x - v\|_1,$ to measure how close $x$ is to being $k$-sparse. We use the notation $B_2^n$ and $B_1^n$ to refer to the unit Euclidean ball and the unit $\ell_1$ ball in $\R^n$, respectively. In what follows $A \lesssim B$ means that there is a universal constant $c$ such that $A \leq cB$, and $\gtrsim$ is defined analogously. 
We use the notation $\EE_g$ to indicate expectation taken with respect to the random variable $g$ (i.e., conditional on all other random variables). The operator norm of a matrix $A$, $\|A\|_{p\to q}$, is defined as $\|A\|_{p\to q} := \max\limits_{\|x\|_p=1} \|Ax\|_q,$ though we may simply use $\|A\|$ in place of $\|A\|_{2\to 2}$.

\subsection{Roadmap}

 The rest of the paper is organized as follows.  In Section \ref{sec:contributions} we highlight the main contributions of the paper, and in Section \ref{sec:math_prelim} we review some necessary technical definitions and tools from compressed sensing, embedding theory, and probability. Section \ref{sec:quantization} is dedicated to a review of quantization theory, and in particular the noise-shaping techniques that are instrumental to our results. Sections \ref{sec:binary} and \ref{sec:CSresults} are dedicated to our main theorems and proofs on binary embeddings and quantized compressed sensing. Section \ref{sec:RIPproof} provides a proof of a critical technical result, namely that certain matrices associated with our methods satisfy a restricted isometry property. Finally Section \ref{sec:technical} provides proofs of some technical lemmas that we require along the way.


\section{Contributions}\label{sec:contributions}
Below, we describe our main results as they pertain to binary embeddings and quantized compressed sensing, and we summarize the main technical contribution that allows us to obtain our results. To that end, let $D_{\epsilon}$ be an $n\times n$ diagonal matrix with random signs, $\Phi$ an $m\times n$  matrix as in Construction \ref{distribution}, $\Q$ be a noise-shaping quantizer, i.e., either $\Q_{\Sigma \Delta}$ a $\Sigma\Delta$ quantization or $\Q_{\beta}$ a distributed noise-shaping quantization (as in Sections \ref{sec:Sigma Delta} and \ref{sec:distributed}). Let $\widehat{V}$ be the associated linear operator defined in \eqref{Vtilde} below and $\widetilde{V} = (9/8)\widehat{V}.$ 


\subsection{Contributions related to binary embeddings}
 \begin{algorithm}[tb]
\caption{Fast binary embeddings}
\label{alg:binary_embedding}
\begin{algorithmic}[1]
    \REQUIRE  \hfill \vspace{-5pt}
\begin{itemize}\itemsep -5pt
\item[(i)] a matrix $\Phi\in \R^{m\times n}$ drawn according to Construction \ref{distribution}, 
\item[(ii)]a stable noise-shaping quantization scheme $\mathcal{Q}$ ($\Sigma\Delta$, as in Section \ref{sec:Sigma Delta} or "distributed", as in Section \ref{sec:distributed}) with alphabet $\mathcal{A}=\{\pm 1\}$ and associated matrix $\widetilde{V}\in \R^{p\times m}$ as in Definition \ref{SDcond} or \ref{Beta_cond},
\item[(iii)]a diagonal matrix $D_\epsilon \in \R^{n\times n}$ with random signs,
\item[(iv)] any points $a, b \in \mathcal{T}\subset{\R^n}$ with $\|a\|_1 \leq 1, \|b\|_1 \leq 1$.
\end{itemize}
    \ENSURE $d_{\widetilde{V}}(a,b) \approx \|a-b\|_2$.
    \STATE Compute $f(a) = \mathcal{Q} (\frac{8\Phi D_\epsilon a}{9}) \in \{ \pm 1 \}^m$ and $f(b) = \mathcal{Q} (\frac{8\Phi D_\epsilon b}{9}) \in \{\pm 1\}^m$.
    \STATE  Return $d_{\widetilde{V}}(a,b) = \| \widetilde{V}(f(a)-f(b)) \|_2$.
    \end{algorithmic}
\end{algorithm}

Let $\mathcal{T}\subset \R^n$ be a subset of the unit $\ell_1$-ball. In Algorithm \ref{alg:binary_embedding} we construct quasi-isometric (i.e., distance preserving) embeddings between  $(\mathcal{T}, \|\cdot\|_2)$ and the binary  cube $\{-1,+1\}^m$ endowed with the pseudometric
$$d_{\widetilde{V}}(\tilde{q},q) := \| \widetilde{V}(\tilde{q} -{q}) \|_2.$$

As they rely on Construction \ref{distribution} our embeddings $f: \mathcal{T} \to \{\pm 1\}^m$  support fast computation, and despite their highly quantized non-linear nature, they perform as well as linear Johnson-Lindenstrauss methods up to an additive error that decays exponentially (or polynomially) in $m$. \smallskip

{\bf \noindent When $\mathcal{T}$ is finite} we show that with high probability 
and for prescribed distortion $\alpha$
\begin{align*} m \gtrsim \frac{\log{(|\mathcal{T}|)} \log^4 n}{\alpha^2} \quad \implies  \quad \big|d_{\widetilde{V}}(f(x), f(\tilde{x})) - \|x - \tilde{x}\|_2\big| & \leq {\alpha}\|x - \tilde{x}\|_2 + c\eta(m), \\ &\quad\text{where} \ \eta(m) \xrightarrow[m \rightarrow \infty]{} 0. \end{align*}
Above, $\eta(m)$ decays polynomially fast in $m$ (when $\mathcal{Q}$ is a $\Sigma\Delta$ quantizer), or exponentially fast (when $\mathcal{Q}$ is a distributed noise shaping quantizer). In fact, we prove a more general version of the above result, which applies to infinite sets $\mathcal{T}$. 

{\bf \noindent When $\mathcal{T}$ is arbitrary (possibly infinite)} we show that with high probability  and for prescribed distortion $\alpha$
\[ m \gtrsim \frac{\log^4 n}{\alpha^2}\cdot\frac{\omega(\mathcal{T})^2}{\rad(\mathcal{T})^2} \quad \implies \quad \big|d_{\widetilde{V}}(f(x), f(\tilde{x})) - \|x - \tilde{x}\|_2\big| \leq \max(\sqrt{\alpha},\alpha) \rad(\mathcal{T})+ c\eta(m)\]
where $\eta(m)$ is as before and where $\omega(\mathcal{T})$ and $\rad(\mathcal{T})$ denote the Gaussian width and Euclidean radius of  $\mathcal{T}$, defined in Section \ref{sec:math_prelim}.

We remark that our results are even more general in the sense that we can embed into any finite set $\mathcal{A}^m$, where $\mathcal{A}$ is the quantization alphabet used by the noise-shaping quantization scheme (see Section \ref{sec:Sigma Delta} for details). In particular, this means that one may choose an alphabet $\mathcal{A}$ that contains more than two elements and obtain improved constants in the additive part (i.e., $c\eta(m)$) of our error bounds. Thus, in effect our results amount to fast quantized Johnson-Lindenstrauss embeddings.
%
%
%
%


 \subsection{Contributions related to the quantization of compressed sensing measurements arising from random structured matrices}
  \begin{algorithm}[tb]
\caption{Fast Quantized Compressed Sensing}
\label{alg:compressed_sensing}
\begin{algorithmic}[1]
    \REQUIRE  \hfill \vspace{-5pt}
\begin{itemize}\itemsep -5pt
\item[(i)]a matrix $\Phi\in \C^{m\times n}$ drawn according to Construction \ref{distribution}, 
\item[(ii)] A noise-shaping quantization scheme $\mathcal{Q}$ (as in Sections \ref{sec:Sigma Delta} or \ref{sec:distributed}) with associated matrix $\widehat{V}$ as in \eqref{Vtilde}, quantization alphabet $\A$, and appropriate error parameter $\eta$.
\item[(iii)] any point $x \in \Sigma_k^n$ with $\|x\|_2 \leq 1$.
\end{itemize}
    \ENSURE Accurate reconstruction from quantized compressed sensing measurements 
    \STATE Quantize the compressed sensing measurements: $q = \mathcal{Q}(\Phi x)$.
    \STATE Reconstruct an approximation 
    \begin{equation}\label{l1-min variant}
    \hat{x} :=\arg\min_z \|z\|_1  \text{\quad subject to \quad} \|\widehat{V}\Phi z-\widehat{V}q \|_2 \leq \eta.
    \end{equation}
    \end{algorithmic}
\end{algorithm}

 We provide the first non-trivial quantization results, with optimal (linear) dependence on sparsity, in the setting of compressed sensing with bounded orthonormal systems. Our results also apply to partial circulant ensembles and our approach is summarized in Algorithm \ref{alg:compressed_sensing}. 

 We show that noise-shaping quantization approaches, namely $\Sigma\Delta$ quantization (e.g., \cite{GLPSY2013, KSW2012, CGKRO2015}) and distributed noise-shaping quantization \cite{Chou2013} significantly outperform the naive (but commonly used) quantization approach, memoryless scalar quantization.  They yield polynomial and exponential error decay, respectively, as a function of the number of measurements. In contrast, the error decay of memoryless scalar quantization can be no better than linear \cite{GVT1998, BJKS2015}. 

Noise shaping quantization techniques produce vectors $q$ whose entries belong to a finite set. Equally importantly,  $q$ is related to the vector of measurements $y=Ax$ via a noise-shaping relation of the form 
\[ y - q = H u, \]
where $H$ is a so-called noise shaping operator and $u$ is a vector with bounded entries. 
Recalling that in compressed sensing, one wishes to recover the (approximately) sparse vector $x$, our reconstruction approach is to solve the optimization problem 
\[ \min_z \|z\|_1  \text{\quad subject to \quad} \|\widehat{V}\Phi z-\widehat{V}q \|_2 \leq \eta. \]
Above, we choose the matrix $\widehat{V}$ and we set the scalar $\eta$ depending on the quantization scheme. As a result we can show that 
\[ 
m \gtrsim k \log^4 n  \quad \implies \quad 	\|x -  \hat{x}\|_2 \lesssim  \eta(m) +  \frac{\sigma_k(x)_1}{\sqrt{k}} 
\]
with high probability, for all $x\in \R^n$. Above, the estimate $\hat{x}$ is produced by Algorithm \ref{alg:compressed_sensing} and as before $\eta(m)$ decays polynomially fast in $m$ (when $\mathcal{Q}$ is a $\Sigma\Delta$ quantizer), or exponentially fast (when $\mathcal{Q}$ is a distributed noise shaping quantizer). 

\subsection{Technical considerations}
The choice of $\widehat{V}$ is critical to the success of our reconstruction, as we require $\widehat{V}\Phi$ to satisfy a restricted isometry property \cite{CTR2006, CT2006}, see Section \ref{sec:quantization}. Once $\widehat{V}\Phi$ satisfies this property, the reconstruction step of Algorithm \ref{alg:compressed_sensing} yields an estimate ${\hat{x}}$ with $\|\hat{x}-x\|_2 \leq C\eta$ by standard compressed sensing results (see, e.g. \cite{FR2013}). Thus, the main technical challenge is to construct $\widehat{V}$ in such a way that $\widehat{V}\Phi$ satisfies the restricted isometry property. This is non-trivial due to the dependencies implicit in the random model which generates bounded orthonormal systems and partial circulant matrices. Simultaneously, the scalar $\eta$ that we choose must be small enough to yield the desired polynomial or exponential decay rates in the number of measurements $m$. The choice of $\widehat{V}$ that achieves our goals is 
\begin{equation}\label{Vtilde}
\widehat{V} = \frac{1}{\|v\|_2 \sqrt{p}} (I_p \otimes v), 
\end{equation}
where $v = (v_1, \ldots, v_{\lambda})\in \R^\lambda$ depends on the quantization scheme  ($v$ is given explicitly in Sections \ref{sec:Sigma Delta} and \ref{sec:distributed}), $I_p$ is the $p\times p$ identity matrix and $\lambda p = m$. Here, $\otimes$ denotes the Kronecker product. We show that with high probability, for matrices $\Phi$ drawn according to Construction \ref{distribution}  
\[ m \gtrsim \frac{k \log^4 n}{\alpha^2}  \quad \implies \quad \sup_{x \in T_k}\left|  \|\widehat{V}\Phi x\|_2^2 - \|x\|_2^2\right| \leq \alpha, \]
where $T_k$ denotes the set of $k$-sparse signals on the unit sphere. That is, we show that $\widehat{V}\Phi$ has the restricted isometry property. This restricted isometry property is also crucial for our results on binary embeddings.

\section{Definitions and Tools}\label{sec:math_prelim}
\subsection{The restricted isometry property and its implications}
We now recall the restricted isometry property of matrices along with some of its implications to compressed sensing, as well as theorems showing that BOEs and PCEs satisfy it with high probability if $m\gtrsim k \operatorname{polylog}n$.
\begin{definition}[The restricted isometry property (RIP)]\label{def:RIP}
We say that a matrix $A\in \C^{m\times n}$ satisfies the $(\delta_k,k)$-RIP if for every $x\in \Sigma_k^n$, we have 
\[ (1-\delta_k)^2\|x\|^2\leq  \|Ax \|_2^2 \leq (1+\delta_k)^2\|x\|^2.  \]
\end{definition}

It is by now well-known that, if a matrix $A$ satisfies $(\delta_k, k)$-RIP with an appropriate small constant $\delta_k$, $\ell_1$-minimization is robust. We present one such result as a representative.

\begin{theorem}[See Theorem 6.12 in \cite{FR2013}]\label{FR Theorem}
Suppose that the matrix $A\in\C^{m\times n}$ satisfies the $(\delta_{2k},k)$-RIP with $\delta_{2k}(A) < \frac{4}{\sqrt{41}}.$ Then for any $x\in\C^n$ and $b\in\C^m$ with $\|A x - b\|_2 \leq \eta,$ a solution $x^\#$ of 
\[
	\min_{z\in\C^n} \|z\|_1 \mbox{ subject to } \|A z - b\|_2 \leq \eta
\]
approximates the vector $x$ with 
\begin{align*}
\|x - x^\#\|_2 &\leq C\eta + \frac{D}{\sqrt{k}}\sigma_k(x)_1,
\end{align*}
where the constants $C, D > 0$ depend only on $\delta_{2k}.$
\end{theorem}

Indeed several ensembles of random matrices satisfy the RIP and in this paper we are particularly interested in bounded orthonormal ensembles and partial circulant matrices for which we recall the following results. 

\begin{theorem}[See, e.g., \cite{RV2008, FR2013}]\label{thm:RV}
Let $A\in\C^{m\times n}$ be drawn from the bounded orthogonal ensemble. Fix $\alpha \in (0,1).$ Then there exists a universal constant $C >0$ such that
\begin{equation}\label{BOS-RIP}
	\EE \sup_{x\in T_k} \left| \frac{1}{m}\|Ax\|_2^2 - \|x\|_2^2 \right| \leq \alpha,
\end{equation}
provided $\frac{m}{\log(9m)} \geq C\frac{k\log^2(4k)\log(8n)}{\alpha^2}$.
\end{theorem}

\begin{theorem}[See, e.g., \cite{KMR2014}]\label{thm:RIP-KMR}
Let $A\in\R^{m\times n}$ be drawn from the partial circulant ensemble associated with an index set $\Omega$ of size $m$. Fix $\alpha \in (0, 1)$. Then there exists a universal constant $C > 0$ such that
\begin{equation}\label{Circulant-RIP}
	\EE \sup_{x\in T_k} \left| \frac{1}{m}\|Ax\|_2^2 - \|x\|_2^2 \right| \leq \alpha,
\end{equation}
provided $m \geq C\frac{k\log^2(k)\log^2(n)}{\alpha^2}$.
\end{theorem}

\subsection{Embedding tools}
We begin with some useful definitions of  \emph{quasi-isometric} embeddings, and of Gaussian widths of sets.

\begin{definition}[see, e.g., \cite{BGBL2008}]
A function $f: \mathcal{X} \to\mathcal{Y}$ is called a \emph{quasi-isometry} between metric spaces $(\mathcal{X}, d_{\mathcal{X}})$ and $(\mathcal{Y}, d_{\mathcal{Y}})$ if there exist $C > 0$ and $D \geq 0$ such that
\[
	\frac{1}{C}d_{\X}(x, \tilde{x}) - D \leq d_{\mathcal{Y}}(f(x), f(\tilde{x})) \leq Cd_{\X}(x,\tilde{x}) + D
\]
for all $x, \tilde{x} \in \X.$
\end{definition}
\begin{remark} We will relax the above definition by only requiring that $d_\mathcal{Y}$ be a pseudo-metric, so we require $d_\mathcal{Y}(x,y)=d_\mathcal{Y}(y,x)$ and we require $d_\mathcal{Y}$ to respect the triangle inequality, but we allow $d_\mathcal{Y}(x,y)=0$ for $x\neq y$. 
\end{remark}

\begin{definition}\label{def: pseudometric}
 Let $\mathcal{C} = \{-1,1\}^m$ be a discrete cube in $\R^m$ and let $V\in \R^{p\times m}$ with $p\leq m$. We define a \emph{pseudo-metric} $d_{V}$ on $\mathcal{C}$ by \[ d_{V}(q, \tilde{q}) = \|V(q - \tilde{q})\|_2.\] 
\end{definition}

\begin{definition}\label{def:gaussian_width}
For a set $\mathcal{T} \subset \R^n$, the \emph{Gaussian mean width} $\omega(\mathcal{T})$ is defined by
\[
	\omega(\mathcal{T}) := \EE[\sup_{v\in\mathcal{T}} \langle v, g\rangle]
\]
where $g$ is a standard Gaussian random vector in $\R^n.$ Moreover, \[\rad(\mathcal{T}) = \sup_{v\in \mathcal{T}}\|v\|_2\] is the maximum Euclidean norm of a vector in $\mathcal{T}$.
\end{definition}

Krahmer and Ward \cite{KW2011} showed that matrices satisfying the RIP can be used to create Johnson-Lindenstrauss, i.e., near-isometric, embeddings of finite sets. 
\begin{theorem}[\cite{KW2011}]\label{thm:KrahmerWard}
Fix $\nu > 0$ and $\alpha \in (0, 1)$. Consider a finite set $\T \subset\R^n.$ Set ${k \geq 40(\log(40|\T|) + \nu),}$ and suppose that $\Psi\in\R^{p\times n}$ satisfies $(\alpha/4, k)$-RIP. Let $D_\epsilon \in\R^{n\times n}$ be a random sign diagonal matrix. Then, with probability exceeding $1 - e^{-\nu}$, the matrix $\Psi D_\epsilon$ satisfies
\[ 
	\big|\|\Psi D_\epsilon\, x\|_2^2 - \|x\|_2^2\big| \leq \alpha\|x\|_2^2
\]
simultaneously for all $x\in\T$.
\end{theorem}

We will need the following ``multiresolution'' version of the restricted isometry property, used by Oymak et. al in \cite{ORS2015} to prove embedding results for arbitrary, i.e., not necessarily finite, sets.
\begin{definition}[Multiresolution RIP \cite{ORS2015}]
Let $L = \lceil \log_2 n\rceil$. Given $\delta > 0$ and a number $s \geq 1,$ for $\ell = 0, 1, \ldots, L,$ let $(\delta_{\ell}, s_{\ell}) = (2^{\ell/2}\delta, 2^\ell s)$ be a sequence of distortion and sparsity levels. We say a matrix $A\in \R^{m\times n}$ satisfies the Multiresolution Restricted Isometry Property (MRIP) with distortion $\delta > 0$ at sparsity $s$, if for all $\ell \in \{1, 2, \ldots, L\},$ $(\delta_{\ell}, s_{\ell})$-RIP holds.
\end{definition}

\begin{theorem}[\cite{ORS2015}]\label{thm:OymakRechtSoltanolkotabi}
Let $\mathcal{T} \subset \R^n$ and suppose the matrix $\Psi\in\R^{m\times n}$ obeys the Multiresolution RIP with sparsity and distortion levels 
\[
	s = 150(1 + \eta) \quad\text{and}\quad \tilde{\alpha} = \frac{\alpha\cdot\rad(\T)}{C\max\{\rad(\mathcal{T}), \omega(\mathcal{T})\}},
\]
with $C > 0$ an absolute constant. Then, for a diagonal matrix $D_\epsilon$ with an i.i.d. random sign pattern on the diagonal, the matrix $A = \Psi D_\epsilon$ obeys
\[ \sup_{x\in\mathcal{T}}\left|\|Ax\|_2^2 - \|x\|_2^2\right| \leq \max\{\alpha, \alpha^2\}(\rad(\mathcal{T}))^2,\]
with probability at least $1 - e^{-\eta}$.
\end{theorem}

\subsection{Tools from probability}
We start with a definition of the $\gamma_2$ functional, followed by a result by Krahmer et al. \cite{KMR2014} that will be very useful in proving our technical results.
 
\begin{definition}[\cite{talagrand2006generic}, also see  \cite{KMR2014}]\label{def:gamma_2}
For a metric space $(T,d)$, an admissible sequence of $T$ is a collection of subsets of $T$, $\{T_r:r\geq 0\}$, such that for every $s\geq 1$, $|T_r|\leq 2^{2^r}$ and $|T_0|=1$. For $\beta\geq1$, the $\gamma_\beta$ functional is given by the following infimum, taken over all admissible sequences of $T$:
\[\gamma_\beta(T, d) = \inf\sup_{t\in T} \sum_{r = 0}^\infty 2^{r/\beta} d(t, T_r).\]  

\end{definition}
\begin{theorem}[See, e.g., \cite{KMR2014}]\label{thm:KMR}
 Let $\mathcal{M} \subset\mathbb{C}^{p\times n}$ be a symmetric set of matrices, i.e., $\mathcal{M} = -\mathcal{M}.$ Let $\epsilon$ be a $\pm 1$ Bernoulli vector of length $n$. Then
$$\EE \sup_{M\in\mathcal{M}} \left| \|M\epsilon\|_2^2 - \EE\|M\epsilon\|_2^2 \right| \lesssim d_F(\mathcal{M})\gamma_2(\mathcal{M}, \|\cdot\|_{2\to 2}) + \gamma_2(\mathcal{M}, \|\cdot\|_{2\to 2})^2,$$
where $d_F(\mathcal{M}) = \sup_{M\in\mathcal{M}} \|M\|_F$ and $\gamma_2(\cdot)$ is as in Definition \ref{def:gamma_2}.
\end{theorem}
Next we present a version of Bernstein's inequality that we will use in proving that our matrix $\widehat{V}\Phi$ satisfies an appropriate RIP with high probability.

\begin{theorem}[Theorem 8.42 in \cite{FR2013}]\label{thm:probability}
Let $\mathcal{F}$ be a countable set of functions $F:\C^n \to \R$. Let $Y_1, \ldots, Y_M$ be independent random vectors in $\C^n$ such that $\EE F(Y_{\ell}) = 0$ and $F(Y_{\ell}) \leq K$ almost surely for all $\ell \in [M]$ and for all $F \in \mathcal{F}$ with some constant $K > 0.$ Introduce \[ Z = \sup_{F\in \mathcal{F}} \sum_{\ell = 1}^M F(Y_\ell). \]
Let $\sigma_{\ell}^2 > 0$ such that $\EE[F(Y_{\ell})^2] \leq \sigma_{\ell}^2$ for all $F\in \mathcal{F}$ and $\ell\in[M]$. Then, for all $t > 0,$
\[ \PP(Z \geq \EE Z + t) \leq \exp\left(-\frac{t^2/2}{\sigma^2 + 2K\EE Z + tK/3}\right),\]
where $\sigma^2 = \sum_{\ell = 1}^M\sigma_{\ell}^2.$
\end{theorem}

\section{Quantization: Background and preliminaries}\label{sec:quantization}


As previously mentioned, various choices of quantization maps and reconstruction algorithms have been proposed in the context of binary embedding and compressed sensing. 
Below, we first review the \emph{lower bounds} associated with optimal (albeit highly impractical) quantization and decoding, which no scheme can outperform. We then review the basics of \emph{memoryless scalar quantization} as well as \emph{noise shaping quantization}, as we will need these for our results. 

\subsection{Optimal error bounds via vector quantization}

For a class $\mathcal{X}$ of signals (e.g., $\mathcal{X} = \Sigma_k^n$), the performance of a practical data acquisition scheme consisting of three stages (i.e., measurement, quantization, and reconstruction) can be measured against the optimal error bounds. Specifically, the performance of any such a scheme can be measured by the tradeoff between the number of bits (rate) and reconstruction error (distortion) and we now describe how to examine the optimal rate-distortion relationship. Given a desired worst-case  approximation error $\epsilon$, measured in the $\ell_2$ norm,  we can cover the class $\mathcal{X}$ with balls of radius $\epsilon$ associated with the Euclidean norm $\|\cdot\|_2$. The smallest number of such $\epsilon$-balls is called the covering number of $\mathcal{X}$ and denoted by $\mathcal{N}(\mathcal{X}, \|\cdot\|_2, \epsilon).$ An optimal scheme consists of  encoding each signal $x$ in the class $\mathcal{X}$ using $R := \log_2\mathcal{N}(\mathcal{X}, \|\cdot\|_2, \epsilon)$ bits by mapping $x$ to the center of an $\epsilon$-ball in which it lies. A simple volume argument yields the optimal lower bound of the approximation error in term of the rate (see, e.g. \cite{BB2007, BJKS2015}) which, when $\mathcal{X}=\Sigma_k^n$, satisfies $$ \epsilon(R) \gtrsim \frac{n}{k}\,2^{-R/k}.$$

On the other hand, the distortion of a three-stage scheme comprised of measurements using a matrix $A$, quantization via the map $\mathcal{Q}$, and reconstruction using an algorithm $\mathcal{R}$ is defined by \[\mathcal{D}(A,\Q, \mathcal{R}) := \sup_{x\in\mathcal{X}}\|x - \mathcal{R}(\Q(Ax))\|_2.\] 
Thus, trivially, 
$$\mathcal{D}(A, \Q, \mathcal{R})\gtrsim \frac{n}{k}2^{-R/k},$$
and one seeks practical\footnote{Note that the optimal scheme described herein is impractical as it requires direct measurements of $x$ which may not be available, and it requires finding the ball to which $x$ belongs. As the number of such balls scales \emph{exponentially} with the ambient dimension, this task becomes prohibitively expensive.} schemes that approach this lower bound.

In the case of Gaussian (and subgaussian) matrices $A$, there exist schemes \cite{baraniuk2017exponential, Chou2013, SWY2017_2} that approach the optimal error bounds, albeit with various tradeoffs. Among these, the results in \cite{baraniuk2017exponential} apply only to Gaussian random matrices $A$, while those of 
 \cite{Chou2013} and \cite{SWY2017_2} rely on noise-shaping quantization techniques and apply to a more general class of subgaussians. In the quantization direction, our work extends the noise-shaping results of  \cite{Chou2013} and \cite{SWY2017_2} to the case of structured random matrices which are selected according to Construction \ref{distribution}. Before describing the necessary technical details pertaining to noise-shaping we quickly describe the most basic  (but highly suboptimal) approach to quantization, i.e., memoryless scalar quantization, as the most widely assumed approaches to binary embeddings and quantization of compressed sensing measurements use it.

\subsection{Memoryless Scalar Quantization}
The most intuitive approach to quantize linear measurements $y = Ax$ is the so-called memoryless scalar quantization (MSQ) in which the quantization map $\Q_{\text{MSQ}}$ is defined by
\[
	q_i := (\Q_{\text{MSQ}}( y))_i := \arg\min_{r\in \A}|y_i - r|.
\]
In other words, we round each measurement $y_i$ to the nearest element in the quantization alphabet $\A$. For instance, in the case of binary embedding, i.e., $\A = \{-1, 1\}, (\Q_{\text{MSQ}}(y))_i = \sign(y_i).$ Here $\sign(t) = 1$ if $t \geq 0$ and $-1$ otherwise. 

If the quantization alphabet $\A = \delta \Z$ for some step size $\delta > 0$ is to be used, we can bound  $\|y - q\|_2 \leq \frac{1}{2}\delta\sqrt{m}.$ Thus, in the setting of binary embedding, MSQ defines a \emph{quasi-isometric} embedding
\begin{equation} (1 - \alpha)\|v - w\|_2 - \delta \leq \frac{1}{\sqrt{m}}\|\Q(Av) - \Q(Aw)\|_2 \leq (1 + \alpha)\|v - w\|_2 + \delta \label{eq:err_bin_MSQ}\end{equation}
with high probability if $A$ is a Gaussian random matrix \cite{Jacques2015}.
On the other hand, if we consider the quantization problem in compressed sensing, the robust recovery result \eqref{error bound} guarantees that
\begin{equation}
	\|x - \hat{x}_{\text{MSQ}}\|_ 2 \lesssim \delta \label{eq:err_quant_MSQ}
\end{equation}
when $x$ is a $k$-sparse vector.

In spite of its simplicity, MSQ is not optimal when one oversamples, i.e., one has more measurements than needed. In particular, the error bounds \eqref{eq:err_bin_MSQ} and \eqref{eq:err_quant_MSQ} do not improve with increasing the number of measurements $m$. This is due to the fact that MSQ ignores correlations between measurements as it acts component-wise. While one can in principle still get a reduction in the error by  using a finer quantization alphabet (i.e., a smaller $\delta$), this is often not feasible in practice because the quantization alphabet is fixed once the hardware is built, or not desirable as one may prefer a simple, e.g., 1-bit embedding.  In order to address this problem, noise-shaping quantization methods, such as Sigma-Delta ($\Sigma\Delta$) modulation and alternative decoding methods (see, e.g., \cite{DD2003, Gunturk2003, BPY2006, GLPSY2013, KSW2012, Chou2013, FK2014, W2016, FKS2017, CGKRO2015, BJKS2015, H2016}), have been proposed in the settings of quantization of bandlimited functions, finite frame expansions, and compressed sensing measurements. However, these methods have not been studied in the framework of binary embedding problems. As stated before, the main contributions of our work are to extend noise-shaping approaches to embedding problems and to compressed sensing with structured random matrices.

\subsection{Noise-shaping quantization methods}
Noise-shaping quantizers, for example $\Sigma \Delta$ modulation, were first proposed for analog-to-digital conversion of bandlimited functions (see, e.g., \cite{DD2003, Gunturk2003, ST2005}). Their success is essentially due to the fact they push the quantization error toward the nullspace of their associated reconstruction operators. These methods have been successfully extended to the frameworks of finite frames (see, e.g., the survey \cite{CGKRO2015} and  \cite{BPY2006, BLPY2010, LPY2010, GLPSY2013, KSW2012, KSY2014}) and compressed sensing 
(see, e.g. \cite{GLPSY2013, KSW2012, Chou2013, FK2014, H2016, W2016, FKS2017, CGKRO2015, BJKS2015}). In fact, the approaches based on $\Sigma \Delta$ quantization  \cite{SWY2018} and beta encoders \cite{Chou2013} achieve near-optimal bounds for sub-Gaussian measurements.

To explain these methods, consider a real quantization alphabet $\A_{L, \delta}$ consisting of $2L$ symmetric levels of spacing $2\delta,$ i.e., \begin{equation}\label{eq:alphabet}\A_{L, \delta} := \{a\delta \;|\; a = -2L + 1, \ldots, -1, 1, \ldots, 2L - 1\},\end{equation} and let ${\A := \A_{L, \delta} + i\A_{L, \delta}}$ be a complex quantization alphabet. 
\begin{remark}We  mention the complex case here for the sake of completeness and remark that all our results apply in the complex setting, but after this discussion we will restrict our attention to real valued matrices and vectors.
\end{remark} 
A noise-shaping quantizer $\Q:\C^m\to\A^m$, associated with the so-called noise transfer operator $H$, is defined such that for each $y\in\C^m$ the resulting quantization $q := \Q(y)$ satisfies the noise-shaping relation
\begin{equation}\label{ns relation}
y - q = Hu,
\end{equation}
where $u\in\C^m$ is called the \emph{state vector} and $\|u\|_\infty \leq C$ for some constant $C$ independent of $m$. Here, $H$ is an $m\times m$ lower triangular Toeplitz matrix with unit diagonal. Noise-shaping quantizers do not exist unconditionally because of the requirement that $\|u\|_\infty$ is uniformly bounded in $m$. However, under certain suitable assumptions on $H$ and $\A$, they exist and can be simply implemented via recursive algorithms \cite{CG2017}. For example, the following lemma provides conditions under which the schemes are stable.
\begin{lemma}\label{lem:noise shaping}
Let $\A := \A_{L, \delta} +i\A_{L, \delta}$. Assume that $H = I - \widetilde{H}$ where $\widetilde{H}$ is strictly lower triangular, and $\mu \geq 0$ such that  $$\|\widetilde{H}\|_{\infty\to\infty} + \mu/\delta \leq 2L.$$ Suppose that $|y_j|_* := \max\{|\Re(y_j)|, |\Im(y_j)|\} \leq \mu$ for all $j\in[m].$ For each $s \geq 1,$ let $w_s := y_s + \sum_{j = 1}^{s - 1}\widetilde{H}_{s, s - j}u_{s - j}$,
\begin{equation}\label{quantizer}
	q_s := (\Q(y))_s = \arg\min _{r\in\A} \left|w_s - r\right|_*,
\end{equation}
and
\begin{equation}\label{recursive}
	u_s := w_s - q_s.
\end{equation}
Then the resulting $q$ satisfies the noise shaping relation \eqref{ns relation} with $\|u\|_\infty \leq \sqrt{2}\delta.$ In the case of strictly real $y$, the $\sqrt{2}$ factor can be replaced by $1$. 
\end{lemma}
The proof of the above is simply by induction and can be found, for example,  in \cite{CG2017}.

Our algorithm are inspired by the approach to reconstruction from noise-shaping quantized measurements used in \cite{Chou2013} and \cite{SWY2017_2} (see also \cite{CGKRO2015}). We first introduce a so-called \emph{condensation operator} $V: \C^m \to \C^p$, for some $p \in [m]$, defined by
\begin{equation}\label{condensation}
	V:= I_p\otimes v = \begin{bmatrix}
	v& & &\\
	&v & &\\
	& & \ddots &\\
	& & & v
	\end{bmatrix},
\end{equation}
where $v$ is a row vector in $\R^\lambda$. Here, for simplicity, we suppose that $\lambda = m/p$ is an integer.
In this approach, to embed or recover the signal $x$, we first apply $V$ to the noise-shaping relation \eqref{ns relation} and obtain
\begin{equation}\label{eq:state_equations}
	V\Phi x - Vq = VHu.
\end{equation}
Our goal is to design a $p\times m$ matrix $V$ and an $m\times m$ matrix $H$ such that the norm $\|VHu\|_2$ is exponentially (or polynomially depending on the quantization scheme) small in $m$ and our quantization algorithm is stable. This in turn helps us achieve exponentially and polynomially small errors in the settings of binary embedding and compressed sensing. We now present two candidates for the pair $(V, H)$ based on $\Sigma \Delta$ quantization and distributed noise shaping.

\subsection{Sigma-Delta quantization}\label{sec:Sigma Delta}
The standard $r$th-order $\Sigma \Delta$ quantizer with input $y$ computes a uniformly bounded solution $u$ to the difference equation
\begin{equation}\label{SDeq}
	y - q = D^r u.
\end{equation}
This can be achieved recursively by choosing $q := \Q_{\Sigma\Delta}(y)$ and $u$ as in the expressions \eqref{quantizer} and \eqref{recursive}, respectively, for the matrix $H = D^r$. Here the matrix $D$ is the $m\times m$ first-order difference matrix with entries given by
\[
	D_{ij} = \left\{\begin{array}{ll}
	1& \text{if } i = j,\\
	-1& \text{if } i = j + 1,\\
	0& \text{otherwise}.
	\end{array}
	\right.
\]
%
Constructing \emph{stable} $\Sigma\Delta$ schemes of arbitrary order $r$, i.e., where \[ \|y\|_\infty \leq \mu \implies \|u\|_\infty \leq c(r)\] with a constant $c(r)$ that is independent of dimensions (but dependent on the order) is non-trivial for a fixed quantization alphabet. Nevertheless, there exist several constructions (see \cite{DD2003, Gunturk2003, DGK2011}) with exactly this boundedness property, albeit with different constants $c(r)$. For our results herein, when we assume a fixed alphabet, we use the stable $\Sigma\Delta$ schemes of \cite{DGK2011}. The relevant result for us is the following proposition from \cite{KSW2012} summarizing the results of \cite{DGK2011}\footnote{In the relevant proposition of \cite{KSW2012}, there is a typo (in equation (12)) as there is a missing $r^r$ factor. This can be seen from comparing their equations (11) and (12). We remedy this in the version presented here.}. 

\begin{proposition}({\cite{DGK2011}, see \cite{KSW2012}} Proposition 1)\label{prop:stableSD}
There exists a universal constant $C>0$ such that for any alphabet $\A=\{\pm (2\ell-1)\delta: \ 1 \leq \ell\leq L, \ \ell\in \Z  \}$, there is a stable $\Sigma\Delta$ scheme such that 
\begin{equation}
\|y\|_\infty \leq \mu \implies \|u\|_\infty \leq C{\delta}\left( \left\lceil \frac{\pi^2}{(\cosh^{-1}(2L-\mu/\delta))^2} \right\rceil \frac{e}{\pi}\cdot r   \right)^r.
\end{equation}
\end{proposition}
In particular, this guarantees that even with the 1-bit alphabet, i.e.,  $\mathcal{A}=\{\pm 1\}$ and $\delta=1$, as assumed in binary embeddings, stability can be guaranteed as long as $\|y\|_\infty \leq \mu < 1$ with a corresponding bound 
\begin{equation}\label{eq:stableSD}
\|y\|_\infty \leq \mu < 1 \implies \|u\|_\infty \leq C \cdot c(\mu)^r \cdot r^r.
\end{equation}

Previous works (e.g., \cite{GLPSY2013, SWY2018}) using $\Sigma\Delta$ schemes in compressed sensing achieve a polynomially small error bound (in $m$). However, these approaches assume gaussian or subgaussian random matrices and the proof techniques are not easily extended to the case of structured random matrices. This is in part due to the role that $D^{-r}$, applied to $\Phi$, plays in the associated proofs; one essentially requires $D^{-r}\Phi$ to have a restricted isometry property. Due to the dependencies among the rows of $\Phi$ in the case of BOEs and PCEs, such a property  is difficult to prove. Nevertheless, in \cite{FKS2017}, the authors were first to study PCEs and prove  error bounds that decay polynomially in $m$. The approach proposed in \cite{FKS2017} uses a different sampling scheme to generate the PCE, a different reconstruction method (whose complexity scales polynomially in $m$, versus $p$ in our case), and requires a different proof technique. 

To circumvent the above issue, and generate the first results for BOEs,  we will need the following condensation operator. Let $$\lambda:= m/p=: r\tilde{\lambda} - r + 1$$ for some integer $\tilde{\lambda}$.
\begin{definition}[$\Sigma\Delta$ Condensation Operator]\label{SDcond}
Let $v$ be a \emph{row vector in $\R^{\lambda}$} whose entry $v_j$ is the $j$th coefficient of the polynomal ${(1 + z + \ldots + z^{\tilde{\lambda }- 1})^r}$. Define the $\Sigma\Delta$ condensation operator $V_{\Sigma \Delta}:\C^m \to \C^p$ by
\begin{equation}\label{sigma delta conden}
	V_{\Sigma \Delta}= I_p\otimes v = \begin{bmatrix}
	v& & &\\
	&v & &\\
	& & \ddots &\\
	& & & v
	\end{bmatrix}, 
\end{equation}
	and define its normalized versions
\[\widetilde{V}_{\Sigma\Delta} = \frac{9}{8 \|v\|_2\sqrt{p}}V_{\Sigma\Delta},\] 
and 
\[\widehat{V}_{\Sigma\Delta} = \frac{1}{\|v\|_2\sqrt{p}}V_{\Sigma\Delta}.\] 
\end{definition}
\begin{example}
In the case $r=1$, we have $v=(1,...,1)\in \R^{\lambda},$ while when $r=2$, $v=(1, \ 2,  \hdots, \ \tilde\lambda-1, \ \tilde\lambda, \ \tilde\lambda-1, \hdots, \ 2, \ 1)\in \R^{\lambda}.$
\end{example}
\begin{lemma}\label{lem:bound VD} For a stable $r$th order $\Sigma\Delta$ quantization scheme  we have
\[\|\widetilde{V}D^r\|_{\infty\to 2} \leq (8r)^{r+1}\lambda^{-r + 1/2}.\]

\end{lemma}
The proof of Lemma \ref{lem:bound VD} is given in Section \ref{sec:technical}. 

\subsection{Distributed noise shaping}\label{sec:distributed}
In \cite{CG2016}, Chou and G\"unt\"urk achieved an exponentially small error bound in the quantization  of unstructured random (e.g., Gaussian) finite frame expansions by using a so-called \emph{distributed noise-shaping approach}. Our work in this paper  extends this approach handle the practically important cases of compressed sensing and binary embeddings, both with structured random matrices drawn from a BOE or a PCE. To that end, we now review distributed noise-shaping. 

Fixing $\beta > 1$, let $H_\beta$ be a $\lambda\times\lambda$ noise transfer operator given by 
\[
	(H_\beta)_{ij} = \left\{\begin{array}{ll}
	1& \text{if } i = j,\\
	-\beta & \text{if } i = j + 1,\\
	0& \text{otherwise}.
	\end{array}
	\right.
\]
Moreover, let $H:\C^m\to\C^m$ be the block-diagonal \emph{distributed noise-shaping transfer operator} given by
\begin{equation}
	H = I_p\otimes H_\beta = \begin{bmatrix}
	H_\beta& & &\\
	&H_\beta & &\\
	&  &\ddots &\\
	& & & H_\beta
	\end{bmatrix}
\label{nsoperator}
\end{equation}
In analogy with Definition \ref{SDcond} we now introduce the distributed noise-shaping condensation operator.

\begin{definition}[Distributed Noise-Shaping Condensation Operator]\label{Beta_cond}
Define the row vector \[v_\beta :=  [\beta^{-1} \; \beta^{-2} \quad\ldots\quad \beta^{-\lambda}] \] and  define the distributed noise-shaping condensation operator $V_{\beta}:\C^m \to \C^p$ by
\begin{equation}
	V_\beta= I_p\otimes v_{\beta} = \begin{bmatrix}
	v_\beta& & &\\
	&v_\beta & &\\
	& & \ddots &\\
	& & & v_\beta
	\end{bmatrix}.
\label{distributed condensation}
\end{equation}
	Define its normalized versions
\[\widetilde{V}_\beta = \frac{9}{8\|v_\beta\|_2\sqrt{p}}V_\beta,\]
and
\[\widehat{V}_\beta = \frac{1}{\|v_\beta\|_2\sqrt{p}}V_\beta.\]
\end{definition}
Notice that in this setting, $$V_\beta H = I_p\otimes (v_\beta H_\beta)$$ and consequently
\begin{equation}\label{bound VH}
\|V_\beta H\|_{\infty\to\infty} = \beta^{-\lambda} \implies \|\widetilde{V}_\beta H\|_{\infty\to 2} \leq \frac{9\beta^{-\lambda}}{8\|v_\beta\|_2} \leq \frac{9\beta^{-\lambda + 1}}{8} 
\end{equation}

since $v_\beta H_\beta = [0\; 0\ldots \beta^{-\lambda}]$ and $\|v_\beta\|_2 \geq \beta^{-1}.$ 


\section{Main results on binary embeddings} \label{sec:binary}
We are now ready to present our main results on binary (and more general) embeddings, showing that Algorithm \ref{alg:binary_embedding} yields quasi-isometric embeddings that preserve \emph{Euclidean} distances. We start with the case of finite sets, and we then present an analysis of our methods for infinite sets.  

\begin{theorem}[{\bf Binary embeddings of finite sets}]\label{thm: bin embed finite}
Consider a finite set 
$\T \subset B_1^n.$
Fix $\alpha \in (0, 1)$  and $\rho\in(0,1)$. 
Let $\Phi \in \R^{m\times n}, D_\epsilon\in \R^{n\times n}, \widetilde{V}\in \R^{p\times m}$ be as in Algorithm \ref{alg:binary_embedding} and suppose $\lambda = m/p$ is an integer. Define $\widetilde{\Phi} = \frac{8}{9}\Phi$ and let $\Q: \R^m \to \{\pm 1\}^m$, be the stable quantization scheme  corresponding either to the $r$th order $\Sigma\Delta$ quantization or to the distributed noise shaping quantization with $\beta\in(1, 10/9],$ that is used in Algorithm \ref{alg:binary_embedding}. 

Fix $\nu > 0$, and define the embedding map $f:\T \to \{\pm 1\}^m$ by $f = \Q\circ\widetilde{\Phi}\circ D_\epsilon$. Assume that ${k \geq 40(\log(8|\T|) + \nu)}$. Then there exist constants, $C_1>0$ and $C_2 > 0$,  that depend on the quantization such that if
\begin{equation}\label{eq:lower bound on p}
m \geq p \geq C_1 \frac{k\log^4 n}{\alpha^2\rho^2},
\end{equation}
the map $f$ satisfies
\[ \big|d_{\widetilde{V}}(f(x), f(\tilde{x})) - \|x - \tilde{x}\|_2\big| \leq {\alpha}\|x - \tilde{x}\|_2 + C_2\eta(\lambda), \quad \forall x,\tilde{x}\in\T\]
with probability exceeding $1-\rho - e^{-\nu}.$ 
Here, $\eta(\lambda) = \lambda^{-r + 1/2}$ 
if we use the $\Sigma\Delta$ quantization scheme and $\eta(\lambda)= \beta^{-\lambda+1}$ 
if we use the distributed noise-shaping scheme.
\end{theorem}


\begin{proof}
The proof follows from Theorem \ref{thm:KrahmerWard} and our main technical result, Theorem \ref{main theorem}, respectively showing that RIP matrices with randomized column signs provide Johnson-Lindenstrauss embeddings, and that the matrices $\widetilde{V}\widetilde{\Phi}$ satisfy the appropriate RIP. The proof also requires bounds on $\|\widetilde{V}H\|_{\infty \to 2}$ to obtain the advertised decay rates. Indeed, 
 we have
\begin{align}
\big|d_{\widetilde{V}}(f(x), f(\tilde{x})) - \|x - \tilde{x}\|_2\big| \leq  \big|d_{\widetilde{V}}(f(x), f(\tilde{x})) & - \|\widetilde{V}\widetilde{\Phi} D_\epsilon\, (x - \tilde{x})\|_2\big| \notag \\ &+ \big| \|\widetilde{V}\widetilde{\Phi} D_\epsilon\, (x - \tilde{x})\|_2 - \|x - \tilde{x}\|_2\big|
\label{eq:summands}\end{align}
so we must bound the two summands on the right to obtain our result and we start with the second term. 

By \eqref{bound expected} in  Theorem \ref{main theorem}, with $\alpha\rho$ in place of $\alpha$, and Markov's inequality, for any $\rho\in(0,1)$, $\widetilde{V}\widetilde{\Phi}$ satisfies $(\alpha, k)$-RIP with probability exceeding $1-\rho$, provided that $p \geq C_1k\log^4 n/(\rho^2\alpha^2)$ for some constant $C_1 > 0$. Hence by Theorem \ref{thm:KrahmerWard}, with probability at least $1-\rho - e^{-\nu}$, the matrix $\widetilde{V}\widetilde{\Phi} D_\epsilon$ satisfies
\[ 
	\big|\|\widetilde{V}\widetilde{\Phi} D_\epsilon (x - \tilde{x})\|_2^2 - \|x - \tilde{x}\|_2^2\big| \leq {\alpha}\|x - \tilde{x}\|_2^2,
\]
which yields
\begin{equation}\label{eq:summand1}
	\big|\|\widetilde{V}\widetilde{\Phi} D_\epsilon (x - \tilde{x})\|_2 - \|x - \tilde{x}\|_2\big| \leq {\alpha}\|x - \tilde{x}\|_2
\end{equation}
simultaneously for all $x, \tilde{x} \in\T.$ 
To bound the first summand in \eqref{eq:summands}, note that in the case of $\Sigma\Delta$ quantization, for any $z\in\mathcal{T}$, \eqref{SDeq} implies that 
\[\|\widetilde{V}f(z) - \widetilde{V}\widetilde{\Phi} D_\epsilon\, z\|_2 = \|\widetilde{V}D^r u_z\|_2 \leq \|\widetilde{V}D^r  \|_{\infty \to 2} \|u_z \|_\infty \leq   C_1 c_2(\mu)^r r^{2r} \lambda^{-r + 1/2},\]
where $u_z$ is the resulting state vector of $z$ as in (\ref{SDeq}).

To obtain the last inequality above,  use Lemma  \ref{lem:bound VD} to bound  $\|\widetilde{V}D^r \|_{\infty \to 2}$; additionally, 
since all the entries of the matrix $\Phi$ are uniformly bounded by $1$ and $ \|z\|_1 \leq 1$, we obtain ${\|\widetilde{\Phi} D_\epsilon z\|_\infty = (8/9) \| \Phi D_\epsilon z\|_\infty \leq 8/9}$ that allows us to invoke the stability bound \eqref{eq:stableSD} to control $\|u_z\|_\infty$. 
Consequently by the triangle inequality we have
\begin{align}\label{eq:summand2}
 \big| \|\widetilde{V}(f(x)-f(\tilde{x}))\|_2 -\| \widetilde{V}\Phi D_\epsilon\, (x-\tilde{x})\|_2 \big| \leq 2C_1 c_2(\mu)^r r^{2r}\lambda^{-r + 1/2},
\end{align}
which yields the desired result for $\Sigma\Delta$ quantization when combined with \eqref{eq:summands}, and \eqref{eq:summand1}.
The exact same technique yields the advertised bound for distributed noise shaping, albeit we now use \eqref{bound VH} in place of Lemma \ref{lem:bound VD}.



\end{proof}

{\color{black}\begin{remark} (Extension to the $\ell_2$ ball) \label{rem:set T}
Theorem \ref{thm: bin embed finite} can be  modified to apply to finite subsets of the unit Euclidean ball instead of the $\ell_1$ ball.  Indeed, under the condition (\ref{eq:lower bound on p}) and the considered event in Theorem \ref{thm: bin embed finite}, the matrix $\frac{1}{\sqrt{m}}\Phi$ satisfies the $(\alpha, k)$-RIP   by Theorems \ref{thm:RV} and \ref{thm:RIP-KMR}. This implies $\frac{1}{\sqrt{m}}\|\Phi D_\epsilon z\|_2 \leq \sqrt{1 + \alpha}\|z\|_2 \leq \sqrt{1 + \alpha}$ for all $z\in \T$ by Theorem \ref{thm:KrahmerWard}. Hence, $\|\Phi D_\epsilon z\|_\infty \leq \sqrt{1 + \alpha}\sqrt{m}.$ So scaling the quantization alphabet carefully, e.g., considering $\{\pm \sqrt{(1+\alpha)m}\}$ in place of $\{\pm 1\}$, ensures our quantization schemes are stable with stability constant $ \sqrt{(1+\alpha)m}$ for all $z\in\T$. This yields the same error bound up to the transformation $\eta(\lambda) \mapsto \sqrt{(1+\alpha)m} ~\eta(\lambda)$. The price of this extension is that for the new $\eta(\lambda)$ to be small (e.g., less than one) and decreasing with $m$ one would require $m \gtrsim p^\frac{r-1/2}{r-1}$ in the case of $\Sigma\Delta$ quantization and $m\gtrsim p\log{m}$ in the case of distributed noise shaping.  This extension is also applicable to Theorem \ref{thm: bin embed} below by using Theorem \ref{thm:OymakRechtSoltanolkotabi} in place of Theorem \ref{thm:KrahmerWard}. 

Alternatively, in the case of a finite subset $\T\subset B_2^n$, to avoid the lower bound on $m$   one can use Bernstein's inequality and a union bound over the $m$ rows of $\Phi$ followed by a union bound over all $x\in\mathcal{T}$ (instead of invoking the restricted isometry property and Theorem \ref{thm:RIP-KMR}). This would yield the bound $\|\Phi D_\epsilon x\|_\infty<c\log n$ for all $x\in \mathcal{T}$, with high probability, provided $|\T|$ is at most polynomial in $n$. In turn, scaling the alphabet by $c\log n$ would imply the same bound in Theorem \ref{thm: bin embed finite} up to the transformation $\eta(\lambda) \mapsto c\eta(\lambda)\log n$.
\end{remark}}

\begin{remark}(The condition on $\beta$) Note that the condition $\beta \in (1,10/9]$ is just a byproduct of the convenient normalization of $\widetilde{\Phi}$ that we chose, and the resulting bound on $\|\widetilde{\Phi}D_\epsilon\|_{2\to\infty}$. Different normalizations would have allowed pushing $\beta$ closer to $2$. \end{remark}

\begin{remark}(Root-exponential error decay)\label{rem:root_exp}
In the case of $\Sigma\Delta$ quantization, one can optimize the bound \eqref{eq:summand2} by selecting $r$ as a function of $\lambda$. Choosing the $\Sigma\Delta$ scheme of order $r^*$, where $r^*$ minimizes the function $c_2(\mu)^r r^{2r}\lambda^{-r + 1/2}$ (see, e.g., \cite{Gunturk2003, KSW2012} for a similar detailed calculation) yields a quantization error decay of $exp(-c\sqrt{\lambda})$ for some constant $c$ so that
 \[ \big|d_{\widetilde{V}}(f(x), f(\tilde{x})) - \|x - \tilde{x}\|_2\big| \leq {\alpha}\|x - \tilde{x}\|_2 + C e^{-c\sqrt{\lambda}}, \quad \forall x,\tilde{x}\in\T.\]

\end{remark}

\begin{remark}(Optimizing the dimension $p$)\label{rem:root_exp}
The parameter $\alpha$ in Theorem \ref{thm: bin embed finite} behaves essentially as $\frac{1}{\sqrt{p}}$, and  $\|x-\tilde{x}\|_2 \leq 2$ for points in the unit $\ell_1$ ball. So, one can consider choosing $p$ to improve the bounds $C( \frac{1}{\sqrt{p}} + e^{-c \sqrt{m/p}} )$ and $C( \frac{1}{\sqrt{p}} + \beta^{- {m/p}+1})$ in the case of $\Sigma\Delta$ and distributed noise-shaping quantization respectively. Slightly sub-optimal but nevertheless good choices of $p$ are then $p_{\Sigma\Delta} \approx m/\log^2(m)$
and $p_{dist.} \approx m/\log(m)$, where the $\approx$ notation hides constants. In turn this yields a quantized Johnson-Lindenstrauss error decay of the form $\frac{\log{m}}{\sqrt{m}}$ and $\sqrt{\frac{\log m}{m}}$ for the two schemes, respectively. 
\emph{In other words, quantizing the Johnson-Lindenstrauss Lemma only costs us a logarithmic factor in the embedding dimension.}
\end{remark}

\begin{remark}\label{rem:bigger_set}(Embedding into $\mathcal{A}^m$) It is easy to see from the proof that one could replace the cube $\{\pm 1\}^m$ by the set $\mathcal{A}^m$ generated by the $2L$ level alphabet  in \eqref{eq:alphabet} provided $(2L-1)\delta \geq 1$. In turn, this simply reduces the constant $C_2$ in Theorem \ref{thm: bin embed finite} by a factor of $L$.\end{remark}

\begin{remark}(Unstructured random matrices)\label{rem:Gaussian}
Note that the only requirement on $\widetilde{\Phi}$ in the proof is that $\widetilde{V}\widetilde{\Phi}$ satisfies an appropriate restricted isometry property. It is not hard to see that matrices $\widetilde{\Phi}$ with independent entries drawn from appropriately normalized Gaussian or subgaussian distributions yield such a property. Hence a version of our results holds for such matrices, and we leave the details to the interested reader. 
\end{remark}

We now present our result on embedding arbitrary subsets of the unit ball, and comment that the contents of Remarks \ref{rem:set T}, \ref{rem:root_exp}, \ref{rem:bigger_set}, and \ref{rem:Gaussian} apply to the theorem below as well, with minor modification.
\begin{theorem}[{\bf Binary embeddings of general sets}]\label{thm: bin embed}
Consider the same setup as Theorem \ref{thm: bin embed finite} albeit now with \emph{any} set $\T \subset B_1^n$.  
Let $\gamma = \|v\|_1/\|v\|_2$. Suppose that 
\begin{equation}\label{condition on m}
m\geq p \geq C_1\gamma^2(1 + \nu)^2\log^4 n \frac{\max\left\{1, \frac{\omega^2(\T - \T)}{(\rad(\T - \T))^2}\right\}}{\alpha^2}.
\end{equation}
%
%
%
Then, with probability at least $1 - e^{-\nu}$, the map $f:\mathcal{T} \to \{-1, 1\}^m$ given by $f = \Q \circ\widetilde{\Phi}\circ D_{\epsilon}$  ensures that
\[ \left|d_{\widetilde{V}}(f(x), f(\tilde{x})) - d_{\ell_2}(x, \tilde{x}) \right| \leq \max\{\sqrt{\alpha}, \alpha\}\rad(\T - \T) + C_2\eta(\lambda), \quad\text{for all } x, \tilde{x} \in \T. \]
Here $\eta := \lambda^{-r + 1/2}$ if we use a stable $\Sigma\Delta$ quantization scheme; and ${\eta := \beta^{-\lambda + 1}}$ if we use the distributed noise-shaping scheme.
\end{theorem}

\begin{proof}


Note that as in Theorem \ref{thm: bin embed finite}, our quantization scheme is stable since $\|\widetilde{\Phi} D_\epsilon z\|_\infty \leq 8/9$ for all $z\in \mathcal T$. To start, as in \cite{ORS2015}, let $L = \lceil \log_2 n \rceil$. Fix $k$ and $\alpha$, then by Theorem \ref{main theorem} we have that for any level $\ell \in \{0, 1, \ldots, L\}$ 
\[ \PP \left(\sup_{x\in T_{2^\ell k}}\left|\|\widetilde{V}\widetilde{\Phi} x\|_2^2 - \|x\|_2^2\right| \geq 2^{\ell/2} \alpha\right) \leq e^{-\nu}\]
if $p \gtrsim k(\log^4 n + \gamma^2 \nu)/\alpha^2.$
Notice that by Theorem \ref{main theorem}, coupled with a union bound applied to $L$ levels and a change of variable $\nu \mapsto \nu + \log L$, we have that $\widetilde{V}\widetilde{\Phi}$ satisfies the Multiresolution RIP with sparsity $k$ and distortion ${\alpha} > 0$ with probability exceeding $1 - e^{-\nu}$ whenever
\begin{equation}\label{bound on m}
m \geq p \gtrsim \frac{k[\log^4 n + \gamma^2(\nu + \log(L + 1))]}{\alpha^2}.
\end{equation}
Since $\log(L + 1) \leq \log^4 n$ and $\log^4 n + \nu\leq (1 + \nu)\log^4 n$, we can achieve \eqref{bound on m} by setting
\[
	m \geq p \gtrsim \gamma^2(1 + \nu)\frac{k\log^4 n}{{\alpha}^2}.
\]
Consequently, we may now apply Theorem \ref{thm:OymakRechtSoltanolkotabi} (with $\widetilde{V}\widetilde{\Phi}$ in place of $H$) with 
$$k = 150(1+ \nu) \quad\text{and}\quad \tilde{\alpha} = \frac{\alpha}{C\max\left(1, \frac{\omega(\T - \T)}{\rad(\T - \T)}\right)}$$  
to conclude that 
\[
	\sup_{x\in\T-\T}\left|\|\widetilde{V}\widetilde{\Phi} D_ {\epsilon} x\|_2^2 - \|x\|_2^2 \right| \leq \max\{\alpha, \alpha^2\}(\rad(\T-\T))^2
\]
with probability exceeding $1 - e^{-\nu}.$ Indeed, this happens whenever
\[ m\geq p \geq C_1\gamma^2(1 + \nu)^2\log^4 n \frac{\max\left\{1, \frac{\omega^2(\T-\T)}{(\rad(\T-\T))^2}\right\}}{\alpha^2}.\]

%
%
%
%
%
Conditioning on the events happening simultaneously, we have that for all $x\in \T-\T$
\begin{equation}\label{eq:string}
\left| \frac{\|\widetilde{V}\widetilde{\Phi} D_{\epsilon} x\|_2}{\|x\|_2} - 1 \right|^2 \leq \left| \frac{\|\widetilde{V}\widetilde{\Phi} D_{\epsilon} x\|_2^2}{\|x\|_2^2} - 1 \right| \leq \frac{\max\{\alpha, \alpha^2\}(\rad(\T-\T))^2}{\|x\|_2^2}.
\end{equation}
This yields
\begin{equation}
\left|\|\widetilde{V}\widetilde{\Phi} D_{\epsilon} x\|_2 - \|x\|_2 \right| \leq \max\{\sqrt{\alpha}, \alpha\}\rad(\T-\T) \quad\text{for all } x\in\T-\T.
\end{equation}
Then for any $w$ and $v$ in $\T$, applying the quantization state equations \eqref{ns relation}, we have (as in the proof of Theorem \ref{thm: bin embed finite}) and proceeding as in the proof of that theorem,
\begin{align*}
\big|\|\widetilde{V}\Q(\widetilde{\Phi} D_{\epsilon} w)  -  \widetilde{V}\Q(\widetilde{\Phi} D_{\epsilon} v)\|_2 - &\|w - v\|_2  \big| \\ &= \left|\|(\widetilde{V}\widetilde{\Phi} D_{\epsilon} w - \widetilde{V}\widetilde{\Phi} D_{\epsilon} v) - (\widetilde{V}Hu_w - \widetilde{V}Hu_v)\|_2 - \|w - v\|_2\right|\\
&\leq \left| \|\widetilde{V}\widetilde{\Phi} D_{\epsilon} (w - v)\|_2 - \|w - v\|_2\right| + \left|\widetilde{V}H(u_w - u_v)\right|\\
&\leq \max\{\sqrt{\alpha}, \alpha\}\rad(\T - \T) + C_2\eta.
\end{align*}
\end{proof}

\begin{remark}
In distributed noise-shaping quantization, we can bound $\gamma^2 \leq \frac{2\beta}{\beta-1}=:C_\beta$ if $1<\beta \leq \frac{10}{9},$ so we can set the condition \eqref{condition on m} as $m\geq p \geq C_\beta C_1(1 + \nu)^2\log^4(n) \max\left\{1, \frac{\omega^2(\T - \T)}{(\rad(\T - \T))^2}\right\}\alpha^{-2}$. In $\Sigma\Delta$ quantization, we can bound $\gamma^2 \leq \lambda = m/p.$ Hence, the condition \eqref{condition on m} becomes $m\geq p \geq C_1\sqrt{m}\log^2(n)\max\left\{1, \frac{\omega(\T - \T)}{(\rad(\T - \T))}\right\}\alpha^{-1}.$ Note that this restriction on dimensions in the $\Sigma\Delta$ case, i.e., $ \sqrt{m}\lesssim p\leq m$ is not problematic as one should use $p\approx \frac{m}{\log(m)^2}$ as described in Remark \ref{rem:root_exp}, and this choice satisfies the restriction.

\end{remark}

\section{Main results on compressed sensing with BOE and PCE: Quantizers, decoders, and error bounds}\label{sec:CSresults}
Herein, we present our results on the compressed sensing acquisition and reconstruction paradigm outlined in Algorithm \ref{alg:compressed_sensing}. First, we cover the case when the acquired BOE or PCE measurements $\Phi x$ are quantized using a $\Sigma\Delta$ scheme, and we show polynomial error decay of the quantization error. Second, we present an analogous result with distributed noise shaping quantization, and we show exponential error decay of the quantization error as a function of the number of measurements.

\begin{theorem}[Recovery results for $\Sigma\Delta$ quantization]\label{thm:Sigma Delta}
Let $k$ and $p$ be in $\{1, \ldots, m\}$ such that the sampling ratio $\lambda := m/p$ is an integer. Let $\Phi\in\C^{m\times n}$ and $\widehat{V}\in\R^{p\times m}$ be as in Algorithm \ref{alg:compressed_sensing}.
Denote by $\Q_{\Sigma \Delta}^r$ a stable\footnote{e.g., $\Q^r_{\Sigma\Delta}$ can be either as in Proposition \ref{prop:stableSD} or as in Lemma \ref{lem:noise shaping} with  $2^r + 1/\delta \leq 2L$ and $H=D^r$ in the latter case.} $r$th-order scheme with the alphabet $\A$. For $\nu > 0$, there exists an absolute constant $c$ such that whenever
$$m \geq p \geq c\max\left\{k\log^4 n, \sqrt{km \nu}\right\}$$
 the following holds with probability at least $1 - e^{-\nu}$ on the draw of $\Phi$.

Let $x\in \R^n$ such that $\|\Phi x\|_\infty \leq \mu <1$ 
and  $q := \Q_{\Sigma \Delta}^r(\Phi x)$, then the solution $\widehat{x}$ to \eqref{l1-min variant} satisfies
\[
	\|\widehat{x} - x\|_2 \leq C_1 \left(\frac{m}{p}\right)^{-r + 1/2}\delta + C_2\frac{\sigma_k(x)_1}{\sqrt{k}}
\]
where $C_1$ and $C_2$ are explicit constants given in the proof.
\end{theorem}
\begin{proof}
The proof essentially boils down to finding an upper bound for $\|\widehat{V}D^ru\|_2$ which we do by controlling $\left\|\widehat{V}D^r\right\|_{\infty\to 2}$ and $\|u\|_\infty$ . Indeed, by Lemma \ref{lem:bound VD}
\[\left\|\widehat{V}D^r\right\|_{\infty\to 2} \leq (8r)^{r+1}\lambda^{-r + 1/2}.\]
Moreover, since $\|\Phi x\|_\infty \leq \mu <1$, the quantization scheme is stable, i.e., $\|u\|_\infty \leq c \delta$, where $c$ may depend on $r$. The noise-shaping relation \eqref{SDeq} implies that
\[
	\|\widehat{V}\Phi x - \widehat{V}q\|_2 = \|\widehat{V}D^r u\|_2 \leq (8r)^{r+1}\lambda^{-r + 1/2}\|u\|_{\infty} \leq  C_1(r)\lambda^{-r + 1/2}\delta.
\]
where in the case of noise shaping schemes as in Lemma \ref{lem:noise shaping} we have $C_1(r) = (8r)^{r+1}$. 
On the other hand,  in the case of stable coarse schemes we have $C_1(r)=(8r)^{r+1} C(r)$, with $C(r)$ being the $r$-dependent constant on the right hand side of the bound in Prop. \ref{prop:stableSD}.

We will now apply 
Theorem \ref{main theorem} with $\gamma^2:=\|v\|_1^2/\|v\|_2^2 \leq m/p$, as it is in the case of $\Sigma\Delta$ quantization. Thus we deduce that if $p \geq c\max\{\frac{k\log^4 n}{\alpha^2}, \frac{\sqrt{km \nu}}{\alpha}\}$, $\widehat{V}\Phi$ satisfies $(\alpha, k)$-RIP with probability at least $1 - e^{-\nu}$. Therefore, robustness of the convex program \eqref{l1-min variant} (see Theorem \ref{FR Theorem}) implies that the solution, $\widehat{x}$, to \eqref{l1-min variant} satisfies
\[
	\|\widehat{x} - x\|_2 \leq C_1 \left(\frac{m}{p}\right)^{-r + 1/2}\delta + C_2\frac{\sigma_k(x)_1}{\sqrt{k}}.
\]

\end{proof}

\begin{theorem} [Recovery Bounds for Distributed Noise-shaping Quantization]
Fix  $ \beta \in(1 ,   2L- \mu/\delta) $  and consider the same setup as Theorem \ref{thm:Sigma Delta} albeit now with the distributed noise shaping quantization.
There exist positive constants $c, C,$ and $D$ such that whenever 
$$m \geq p \geq ck(\log^4(n) + \nu)$$
 the following holds with probability at least $1 - e^{-\nu}$ on the draw of $\Phi$.

Let $x\in \R^n$ such that $\|\Phi x\|_\infty \leq \mu < 1$. and let $q := \Q_{\beta}(\Phi x)$ be its resulting distributed noise-shaping quantization as in Lemma \ref{lem:noise shaping}, then the solution $\widehat{x}$ to \eqref{l1-min variant} satisfies
\[
	\|x -  \widehat{x}\|_2 \leq C\beta^{-\frac{m}{p} + 1}\delta + D \frac{\sigma_k(x)_1}{\sqrt{k}}.
\]

\end{theorem}
\begin{proof}
The proof is identical to the one of Theorem \ref{thm:Sigma Delta}.
The noise-shaping relation \eqref{ns relation} implies that 
\[
	\widehat{V}\Phi x - \widehat{V}q = \widehat{V}Hu,
\]
and Lemma \ref{lem:noise shaping} yields $\|u\|_\infty \leq \delta$.
Hence, using the first equality in \eqref{bound VH} we have
\[
	\|\widehat{V}\Phi x - \widehat{V}q\|_2 = \|\widehat{V}Hu\|_2 \leq \sqrt{p}\|\widehat{V}Hu\|_\infty \leq \sqrt{p}\|\widehat{V}H\|_{\infty\to\infty}\|u\|_\infty \leq \delta\beta^{-\lambda+ 1}.
\]
By Theorem \ref{main theorem} with $\gamma = \|v\|_1/\|v\|_2 \leq \frac{2\beta}{\beta-1}$ and $p \geq c\frac{k(\log^4(n) + \nu)}{\alpha^2}$,
$\widehat{V}\Phi =\frac{V\Phi}{\|v\|_2\sqrt{p}}$ satisfies $(\alpha, k)$-RIP with probability $1 - e^{-\nu}$. This completes the proof by using the robustness of the convex program \eqref{l1-min variant} as before.

\end{proof}

\begin{remark}(Noise robustness)
We note that by modifying the reconstruction technique slightly, non-quantization noise can be handled in a robust way, in an analogous manner to the method in \cite{SWY2017_2}. We leave the details to the interested reader. 
\end{remark}

\section{The Restricted Isometry Property of $\widehat{V}\Phi$}\label{sec:RIPproof}
Our main technical theorem shows that a scaled version of $V\Phi$ satisfies the restricted isometry property. Our proof of Theorem \ref{main theorem} below is based on the technique of \cite{NPW2014}, which in-turn relies heavily on \cite{RV2008}.
The rough architecture of the proof is as follows. We first show that the desired RIP property holds in expectation, then we leverage the expectation result and a generalized Bernstein bound to obtain the RIP property with high probability. The proof that the RIP holds in expectation in turn comprises several steps. A triangle inequality shows that the expected value of the RIP constant is bounded by the sum of the expected RIP constant of a BOS or PCE matrix and the expected supremum of a chaos process. To bound the chaos process, we use Theorem \ref{thm:KMR} which requires controlling a Dudley integral and hence certain covering numbers. Here again, we use a technique adapted from the works of Nelson et al. \cite{NPW2014}, and Rudelson and Vershynin \cite{RV2008}.

\begin{theorem}\label{main theorem} 
Fix $\alpha \in(0, 1)$ and $\nu > 0$. Let $k$ and $p$ be in $\{1, \ldots, m\}$ with  $\lambda := m/p$  an integer. Let 
$$V := \,I_p\otimes v \text{ and }\widehat{V} := \frac{1}{\|v\|_2\sqrt{p}}V$$
 for a row vector $v\in\R^\lambda$, and let $\Phi$ be as in Construction \ref{distribution}.  Then there exist universal constants $C_1$ and $C_2$ that may depend on the distribution of $\Phi$ but not the dimensions, such that if $$m \geq p \geq C_1\frac{k\log^4 n}{\alpha^2},$$ we have
\begin{equation}\label{bound expected}
	\EE \sup_{x \in T_k}\left|  \|\widehat{V}\Phi x\|_2^2 - \|x\|_2^2\right| \leq \alpha;
\end{equation}
and if 
\[
m \geq p \geq C_2k\max\left\{\frac{\log^4 n}{\alpha^2}, \frac{\gamma^2 \nu}{\alpha^2}\right\}\quad\text{or}\quad m \geq p \geq C_2\frac{k(\log^4 n + \gamma^2\nu)}{\alpha^2},
\] 
we have

\[
\mathbb{P}(	\sup_{x \in T_k}\left|  \|\widehat{V}\Phi x\|_2^2 - \|x\|_2^2\right| \leq \alpha) \geq 1-e^{-\nu}.
\]
Here, $\gamma := \|v\|_1/\|v\|_2.$
\end{theorem}
\begin{proof}
We begin with the bound on the expectation.\\
{\bf \noindent Step (I):} Let $S_\ell = \{(\ell - 1)\lambda + 1, \ldots, \ell \lambda\}$ for $\ell = 1, \ldots, p.$ Let $\Lambda$ be a diagonal matrix having the random signs $\epsilon_j$ from Step 2 of Construction \ref{distribution} as its diagonal entries and let $a_j$ be the $j$th row of $A$. Then 
\[
V\Phi x = V\Lambda Ax = V\Lambda \begin{pmatrix} \<a_1, x\> \\ \vdots \\ \<a_m, x\>\end{pmatrix}
= V\begin{pmatrix} \epsilon_1\<a_1, x\> \\ \vdots \\ \epsilon_m\<a_m, x\>\end{pmatrix}
= \begin{pmatrix} \sum_{j\in S_1}v_j\epsilon_j\<a_j, x\> \\ \vdots \\ \sum_{j\in S_p}v_{j-(p-1)\lambda}\epsilon_j\<a_j, x\>\end{pmatrix},
\]
and the expected squared modulus, with respect to the random vector $\epsilon = (\epsilon_j)_{j=1}^m$, of an entry of $V\Phi x$ is
\begin{align*}
\EE_\epsilon\left|\sum_{j\in S_\ell} v_{j-(\ell-1)\lambda}\epsilon_j\<a_j, x\>\right|^2 &= \EE_\epsilon \sum_{j, j'\in S_\ell} v_{j-(\ell-1)\lambda} v_{j'-(\ell-1)\lambda} \epsilon_j \epsilon_{j'} \<a_j, x\> \overline{\<a_{j'}, x\>}\\
&= \sum_{j\in S_\ell}v_{j-(\ell-1)\lambda}^2 |\<a_j, x\>|^2.
\end{align*}

Then, defining $A_{\Omega_j}$ as the restriction of $A$ to its rows indexed by $\Omega_j = \{j, j + \lambda, \ldots, j + (p-1)\lambda\}\subset \{1,...,m\}$, for $j = 1, \ldots, \lambda$ we have

\begin{align*}
\EE_\epsilon \|V\Phi x\|_2^2 &= \sum_{\ell = 1}^p \sum_{j\in S_\ell} v_{j-(\ell-1)\lambda}^2 |\<a_j, x\>|^2\\
 &= \sum_{j=1}^\lambda v_{j}^2 \sum_{j'\in\Omega_j} |\<a_{j'}, x\>|^2\\
 &= \sum_{j=1}^\lambda v_j^2 \|A_{\Omega_j} x\|_2^2.
\end{align*}
  By the triangle inequality
\begin{align*}
\left| \frac{1}{\|v\|_2^2 p} \|V\Phi x\|_2^2 - \|x\|_2^2\right|
&\leq \frac{1}{\|v\|_2^2 p}\left| \|V\Phi x\|_2^2 - \sum_{j = 1}^\lambda v_j^2 \|A_{\Omega_j}x\|_2^2\right| + \left| \frac{1}{\|v\|_2^2}\sum_{j = 1}^\lambda v_j^2 (\frac{1}{p}\|A_{\Omega_j}x\|_2^2 - \|x\|_2^2) \right|\\
&\leq \frac{1}{\|v\|_2^2 p}\left| \|V\Phi x\|_2^2 - \sum_{j = 1}^\lambda v_j^2 \|A_{\Omega_j}x\|_2^2\right| + \frac{1}{\|v\|_2^2}\sum_{j = 1}^\lambda v_j^2 \left|\frac{1}{p}\|A_{\Omega_j}x\|_2^2 - \|x\|_2^2 \right|,
\end{align*}
which implies that for $$T_k := \{ x \in \Sigma_k^n, \|x\|_2=1\}$$ we have
\begin{align*}\EE_{A, \epsilon} \sup_{x \in T_k}\left| \frac{1}{\|v\|_2^2 p} \|V\Phi x\|_2^2 - \|x\|_2^2\right| \leq & \frac{1}{\|v\|_2^2 p}\EE_{A, \epsilon} \sup_{x\in T_k} \left|\|V\Phi x\|_2^2 - \EE_\epsilon\|{{V}}\Phi x\|_2^2\right| \\  & + \frac{1}{\|v\|_2^2} \sum_{j=1}^\lambda v_j^2\; \EE_A\sup_{x\in T_k}\left|\frac{1}{p}\|A_{\Omega_j}x\|_2^2 - \|x\|_2^2 \right|.
\end{align*}

%
The second summand above is the expected value of the RIP constant of matrices $A_{\Omega_j}$ drawn from the BOE or PCE, and this quantity is bounded by $\alpha_A$ provided $p \gtrsim k\log^4 n/\alpha_A^2$ by Theorems \ref{thm:RV} and \ref{thm:RIP-KMR}. 
Therefore,
$$\EE_{A, \epsilon} \sup_{x \in T_k}\left| \frac{1}{\|v\|_2^2 p} \|{{V}}\Phi x\|_2^2 - \|x\|_2^2\right| \leq \frac{1}{\|v\|_2^2p}\EE_{A, \epsilon} \sup_{x\in T_k} \left|\|{{V}}\Phi x\|_2^2 - \EE_\epsilon\|{{V}}\Phi x\|_2^2\right| + \alpha_A.$$
 Here, $\alpha_A$ is the RIP constant in either \eqref{BOS-RIP} or \eqref{Circulant-RIP}. 

{\bf \noindent Step (II):} Now, we need to bound $\EE_{A, \epsilon} \sup_{x\in T_k} \left|\|V\Phi x\|_2^2 - \EE_\epsilon\|{{V}}\Phi x\|_2^2\right|$.  

To control this quantity, we will use Theorem \ref{thm:KMR}, which requires a little bit of setting up. To that end, let $$A_j^v = \begin{bmatrix}
v_1 a_{(j-1)\lambda + 1}\\
\vdots\\
v_\lambda a_{j\lambda}\end{bmatrix},$$
and observe that
\begin{align}
{V}\Phi x
&= \begin{bmatrix}
v_1\<a_1, x\> & \ldots & v_\lambda \<a_{\lambda}, x\> & & & & & & &\\
& & & v_1\<a_{1 + \lambda}, x\> & \ldots & v_\lambda \<a_{2 \lambda}, x\> & & & &\\
& & & & & & \ddots & & &\\
& & & & & & & v_1\<a_{1 + (p-1)\lambda}, x\> & \ldots & v_\lambda \<a_{p\lambda}, x\>
\end{bmatrix}
\begin{bmatrix}
\epsilon_1\\
\epsilon_2 \\
\vdots\\
\epsilon_{m}
\end{bmatrix} \nonumber\\
&=: \begin{bmatrix}
(A_1^v x)^* & & & \\
& (A_2^v x)^* & &\\
& & \ddots & \\
& & & (A_p^v x)^*
\end{bmatrix}\begin{bmatrix}
\epsilon_1\\
\epsilon_2 \\
\vdots\\
\epsilon_{m}
\end{bmatrix} \label{eq:Ajv_def}\\
&=: A_x^v \epsilon.\label{eq:Axv_def}
\end{align}
Let $\A^v$ be the set of matrices given by $$\A^v := \{A_x^v, x\in T_k\}.$$ Then, conditioning on $A$, by Theorem \ref{thm:KMR}
\begin{align}
\EE_\epsilon \sup_{x\in T_k} \left|\|{V}\Phi x\|_2^2 - \EE_\epsilon\|{V}\Phi x\|_2^2\right|
&= \EE_\epsilon \sup_{A_x^v \in \A^v}\left| \|A_x^v \epsilon\|_2^2 - \EE_\epsilon \|A_x^v \epsilon\|_2^2 \right| \notag\\
&\lesssim d_F(\A^v)\gamma_2(\A^v, \|\cdot\|_{2\to 2}) + \gamma_2(\A^v, \|\cdot\|_{2\to 2})^2.
\end{align}
Our next step will be taking expectations over $A$, for which we will obtain an upper bound on the $\gamma_2$ functional, and a bound on the expectation of $d_F(\mathcal{A}^v)$. We start with the latter term, for which we observe that $\|A_x^v\|_F^2 = \sum_{j=1}^\lambda v_j^2 \|A_{\Omega_j} x\|_2^2.$ Hence, taking expectations over $A$, we have
\begin{align} 
\EE[d_F(\A^v)] &= \EE \sup_{x\in T_k}\sqrt{\sum_{j=1}^\lambda v_j^2 \|A_{\Omega_j} x\|_2^2}\\
&\leq \sqrt{\sum_{j=1}^\lambda v_j^2\, \EE\sup_{x\in T_k}\|A_{\Omega_j} x\|_2^2}\\
&\leq \sqrt{\sum_{j=1}^\lambda v_j^2\, \left(p\EE\sup_{x\in T_k}\left|\frac{1}{p}\|A_{\Omega_j} x\|_2^2 - \|x\|_2^2\right| + p\EE\sup_{x\in T_k}\|x\|_2^2\right)}\\
&\leq \sqrt{2p}\|v\|_2\label{eq:dF_bound}
\end{align}
where the third inequality is due to the fact the expected value of the RIP constant of matrices $A_{\Omega_j}$ is bounded by 1.

Now to control the $\gamma_2$ term, we note that since
\begin{align*}
\|A_x^v - A_z^v\|_{2\to 2} &= \max_{\|w\|_2 = 1}\|(A_x^v - A_z^v)w\|_2\\
&= \max_{\|w\|_2 = 1} \sqrt{\sum_{j=1}^\lambda |\<(w_i)_{i = 1 + (j-1)\lambda}^{j\lambda}, A_j^v(x - z)\>|^2}\\
&= \max_j \|A_j^v (x-z)\|_2,
\end{align*}
we have $\|A_x^v - A_z^v\|_{2\to 2} = \max_j \|A_j^v(x - z)\|_2 = \|x - z\|_X$ where we have defined the seminorm 
\begin{equation}\label{eq:Xnorm_def}
\|z\|_X := \max_{j\leq p} \|A_j^v z\|_2.
\end{equation}
Hence, \begin{equation}\gamma_2(\A^v, \|\cdot\|_{2\to 2}) = \gamma_2(T_k,\|\cdot\|_X).\label{eq:gamma2_to_gamma2}\end{equation} Then, 
\begin{align}\EE_{A, \epsilon} \sup_{x\in T_k} \left|\|{V}\Phi x\|_2^2 - \EE_\epsilon\|{V}\Phi x\|_2^2\right| \leq \ \EE_A [d_F(\A^v)\gamma_2(T_k, \|\cdot\|_X)] + \EE_A\gamma_2(T_k, \|\cdot\|_X)^2.\label{eq:KMR_bound}\end{align} 
To bound 
{\color{black} $\gamma_2(T_k, \|\cdot\|_X)$, }
we first observe that for any vector $x$
\[
	\|x\|_X \leq \|v\|_2 \max_{1\leq j \leq p}|\langle a_j, x\rangle| =: \|v\|_2 \|x\|_X',
\]
which implies that
\[
	\gamma_2(T_k, \|\cdot\|_X) \leq \|v\|_2 \gamma_2 (T_k, \|\cdot\|_X').
\]
By Dudley's inequality (e.g., \cite{talagrand2006generic}), we have
\begin{equation}\label{eq:Dudley}
\gamma_2(T_k, \|\cdot\|_X') \lesssim \int_0^{\infty} \sqrt{\log \mathcal{N}(T_k, \|\cdot\|_X', u)}\,du.
\end{equation}
Since $\|a_j\|_\infty \leq 1$, the right hand side of (\ref{eq:Dudley}) can be bounded as
\[
	\int_0^{\infty} \sqrt{\log \mathcal{N}(T_k, \|\cdot\|_X', u)}\,du \lesssim \sqrt{k\log p\log n}\log k
\]
(see, e.g., the proof of Lemma 3.8 in \cite{RV2008}).
Thus, we obtain
\begin{equation}
	\gamma_2(T_k, \|\cdot\|_X) \lesssim \|v\|_2\sqrt{k\log p\log(n)}\log k. \label{eq:gamma2_bound}
\end{equation}

\noindent Finally, substituting \eqref{eq:gamma2_bound} and \eqref{eq:dF_bound} into a scaled version of \eqref{eq:KMR_bound}, we obtain
\begin{align*}
\EE\sup_{x\in T_k}\left\lvert\frac{1}{\|v\|_2^2 p}\left(\|{V}\Phi x\|_2^2 - \EE\|{V}\Phi x\|_2^2\right) \right\rvert &\lesssim \frac{1}{\|v\|_2^2p}(\|v\|_2\sqrt{p}) \|v\|_2\sqrt{k\log p\log(n)}\log k\\
&\qquad + \frac{1}{\|v\|_2^2 p} \|v\|_2^2k\log p\log(n)\log^2 k\\
&\lesssim \frac{\sqrt{k\log p\log(n)}\log k}{\sqrt{p}} + \left(\frac{\sqrt{k\log p\log(n)}\log k}{\sqrt{p}}\right)^2\\
&\lesssim \frac{\sqrt{k\log p\log(n)}\log k}{\sqrt{p}}\\
&\lesssim \sqrt{\frac{k\log^4 n}{p}}\\
&\leq \alpha_1,
\end{align*}
provided that 
$n \gtrsim p \gtrsim k\log^4 n/\alpha_1^2$.
Therefore,
\begin{equation}
	\EE_A\EE_\epsilon \sup_{x \in T_k}\left| \frac{1}{\|v\|_2^2 p} \|\Phi x\|_2^2 - \|x\|_2^2\right| \leq \alpha_1 + \alpha_A = \alpha,
\end{equation}
provided that $p \geq C_1k\log^4 n/\alpha^2$.

\noindent{\bf Step (III): The probability estimate.} We will use Theorem \ref{thm:probability} and follow the technique that was used in \cite{FR2013} (Chapter 12) to prove a similar result for BOE's (i.e., without the condensation operator $V$). To that end, let $X^*_\ell$ be the $\ell$-th row of the matrix $\widehat{V}\Phi,$ and define $Q_{k, n} = \bigcup_{S\subset [n], |S| \leq k} Q_{S, n}$ with $Q_{S, n} = \{(z, w): \|z\|_2 = \|w\|_2 = 1, \text{supp}(z), \text{supp}(w)\subset S\}.$ The restricted isometry constant $\delta_k$ satisfies
\begin{align*}
p\delta_k &= \sup_{S\subset[n], |S| = k} \left\| \sum_{\ell=1}^p\left((X_\ell)_S(X_\ell^*)_S - (I_n)_S\right) \right\|_{2\to 2}
= \sup_{(z, w) \in Q_{k, n}}
\left\langle\sum_{\ell = 1}^p (X_{\ell} X_{\ell}^* - I_n)z, w\right\rangle\\
&= \sup_{(z, w)\in Q_{k, n}} \sum_{\ell = 1}^p 
\left\langle (X_{\ell} X_{\ell}^* - I_n)z, w\right\rangle.
\end{align*}
Let $Q^*_{k, n}$ be a dense countable subset of $Q_{k, n}$. Then
\[ p\delta_k = \sup_{(z, w)\in Q^*_{k, n}} \sum_{\ell = 1}^p f_{z, w}(X_{\ell}),\]
where $f_{z, w}(X) = 
\langle(XX^* - I_n)z, w\rangle$. To apply Theorem \ref{thm:probability}, we first check the boundedness of $f_{z, w}$ for $(z, w)\in Q_{k, n}$. That is,
\begin{align*}
|f_{z, w}(X)| &\leq |\langle (XX^* - I_n)z, w\rangle|\\
&\leq \|X_SX_S^* - (I_n)_S\|_{2\to 2}\\
&\leq  \|X_SX_S^* - (I_n)_S\|_{1\to 1}. 
\end{align*}
Recall that any row of $V\Phi/\|v\|_2$ is of the form $\frac{1}{\|v\|_2}\sum_{\ell=1}^\lambda v_{\ell}\phi_{\ell}^*$ where $\phi_{\ell}^*$ is a row of $\Phi$. Since $\|\phi_{\ell}\|_\infty \leq 1$, with $\delta_{s,s'}$ denoting the Kronecker delta function, we now have
\begin{align*}
|f_{z, w}(X)|&\leq \|X_SX_S^* - (I_n)_S\|_{1\to 1} = \max_{s\in S} \sum_{s' \in S}\left|\left(\frac{1}{\|v\|_2}\sum_{\ell = 1}^\lambda v_{\ell}\phi_{\ell}^*(s)\right)\left(\frac{1}{\|v\|_2}\sum_{\ell = 1}^\lambda v_{\ell}\phi_{\ell}^*(s')\right) - \delta_{s, s'}\right|\\
 &\leq \max_{s\in S} \sum_{s' \in S}\left|\frac{\|v\|_1^2}{\|v\|_2^2}\sup_{\ell}\|\phi_{\ell}\|_\infty \sup_{\ell'}\|\phi_{\ell'}\|_\infty + 1\right|
 \leq 2\gamma^2 k,
\end{align*} 
where $\gamma = \|v\|_1/\|v\|_2.$
We also observe that
\[
	\EE|f_{z, w}(X)|^2 \leq \EE|\langle (XX^* - I_n)z, w\rangle|^2 \leq \EE\|(X_S X_S^* - I_n)z\|_2^2 = \EE\|X_S\|_2^2 |\langle X, z\rangle|^2 - 2\EE|\langle X, z\rangle|^2 + 1.
\]
We bound $\|X_S\|_2^2$ and $\EE|\langle X, z\rangle|^2$ as follows:
\begin{align*}
\|X_S\|_2^2 &= \sum_{j\in S}|X_S(j)|^2
= \sum_{j\in S}\left|\frac{1}{\|v\|_2}\sum_{\ell = 1}^\lambda v_{\ell}\phi_{\ell}(j)\right|^2\\
&\leq \sum_{j\in S}\frac{1}{\|v\|_2^2}\|v\|_1^2 \|\phi_{\ell}\|_\infty^2
\leq \gamma^2 k, 
\end{align*}
and
\begin{align*}
\EE|\langle X, z\rangle|^2 &= \frac{1}{\|v\|_2^2}\EE\left|\sum_{\ell = 1}^\lambda v_{\ell} \langle \phi_\ell, z\rangle\right|^2
= \frac{1}{\|v\|_2^2} \left(\EE \sum_{\ell, \ell'}v_{\ell}v_{\ell'} \langle \phi_{\ell}, z\rangle\overline{\langle \phi_{\ell'}, z\rangle} \right)\\
&=  \frac{1}{\|v\|_2^2} \sum_{\ell, \ell'}v_{\ell}v_{\ell'} \EE \langle \phi_{\ell}, z\rangle\overline{\langle \phi_{\ell'}, z\rangle} 
=  \frac{1}{\|v\|_2^2} \sum_{\ell}v_{\ell}^2 \EE |\langle \phi_{\ell}, z\rangle|^2\\
&= \frac{1}{\|v\|_2^2} \sum_{\ell} v_{\ell}^2\|z\|_2^2 
= 1.
\end{align*}
Therefore,
$
	\EE|f_{z, w}(X)^2| \leq \gamma^2 k - 1 < \gamma^2 k.
$
Applying Theorem \ref{thm:probability} with $t = \alpha_2 p$, we obtain
\[
	\PP(\delta_k \geq \alpha + \alpha_2) \leq \PP(p\delta_k \geq \EE p\delta_k + \alpha_2 p) \leq \exp\left(-\frac{(\alpha_2 p)^2/2}{p\gamma^2 k + 4\gamma^2 kp\EE\delta_k + 2\alpha_2 p\gamma^2 k/3} \right),
\]
where $\EE \delta_k \leq C\sqrt{\frac{k\log^4 n}{p}} \leq \alpha < 1,$ by the first part of the theorem. Hence,
\[
\PP(\delta_k \geq \alpha + \alpha_2) \leq \exp\left(-\frac{\alpha_{2}^2 p}{2k\gamma^2}\frac{1}{1 + 4\alpha + 2\alpha_2/3} \right) \leq \exp\left(-c(\alpha)\frac{\alpha_{2}^2 p}{2k\gamma^2}\right)
\]
where $c(\alpha) := (1 + 4\alpha_1 + 2/3)^{-1} \geq (1 + 4 + 2/3)^{-1} = 3/17.$ The above term is less than $e^{-\nu}$ provided that 
\[ p \geq \frac{34}{3}\frac{\gamma^2 k\nu}{\alpha_2^2}.\]
In conclusion, by taking $\alpha_2 = \alpha$ and rescaling, we have that $\delta_k \leq \alpha$ with probability exceeding $1 - e^{-\nu}$ provided that 
\[ p \geq C_2\max\left\{\frac{k\log^4 n}{\alpha^2}, \frac{\gamma^2 k\nu}{\alpha^2}\right\}\quad\text{or}\quad p \geq C_2\frac{k(\log^4 n + \gamma^2\nu)}{\alpha^2}\]
for some absolute constant $C_2$.
\end{proof}
\section{Technical Lemmas and Proofs}\label{sec:technical}

\subsection*{Proof of Lemma \ref{lem:bound VD}}
We first need the following lemma
\begin{lemma}\label{summation by parts}

Let $(v_\ell)_{\ell \in\mathbb{Z}}$ and $(u_\ell)_{\ell\in\mathbb{Z}}$ be any sequences such that $v_{\ell} = 0$ for $\ell \neq 1, \ldots, \lambda$. Then for any integer $t \geq 0,$
\begin{equation}\label{key claim}
	\sum_{\ell = 1}^\lambda v_\ell (\Delta^r u)_{t + \ell} = (-1)^r \sum_{\ell = 1}^\lambda (\Delta^r v)_{\ell + r}u_{t + \ell} + \sum_{j = 1}^r (-1)^j (\Delta^{j - 1} v)_j (\Delta^{r - j} u)_t.
\end{equation}
\end{lemma}
Here $\Delta$ denotes the first order difference operator. That is, for any sequence $(a_j)_{j\in \Z}$,  ${(\Delta a)_j = a_j - a_{j-1}.}$ 
%

\begin{proof}
The proof is by induction. First, let's check the case $r = 1$. 
\begin{align*}
\sum_{\ell = 1}^\lambda v_\ell (\Delta u)_{t + \ell} &= \sum_{\ell = 1}^\lambda v_\ell u_{t+\ell} - \sum_{\ell = 1}^\lambda v_\ell u_{t + \ell -1}\\
&= \sum_{\ell = 1}^\lambda v_\ell u_{t+\ell} - \sum_{\ell = 1}^\lambda v_{\ell + 1}u_{t + \ell} + v_{\lambda + 1} u_{t + \lambda} - v_1u_t\\
&= (-1)\sum_{\ell = 1}^\lambda (\Delta v)_{\ell + 1}u_{t + \ell} + (-1)v_1u_t.
\end{align*}
Here, we used the fact that $v_{\lambda + j} = 0$ for $j \geq 1$.

Suppose that \eqref{key claim} is true. Then
\begin{align*}
\sum_{\ell = 1}^\lambda v_\ell (\Delta^{r + 1} u)_{t + \ell} &= \sum_{\ell = 1}^\lambda v_\ell (\Delta^r(\Delta u))_{t + \ell}\\
&= (-1)^r \sum_{\ell = 1}^\lambda (\Delta^r v)_{\ell + r}(\Delta u)_{t + \ell} + \sum_{j = 1}^r (-1)^j (\Delta^{j - 1} v)_j (\Delta^{r - j + 1} u)_t\\
&= (-1)^r\left[\sum_{\ell=1}^\lambda (\Delta^r v)_{\ell + r}u_{t +\ell} - \sum_{\ell = 1}^\lambda(\Delta^r v)_{\ell + r}u_{t + \ell - 1}\right] + \sum_{j = 1}^r (-1)^j (\Delta^{j - 1} v)_j (\Delta^{r - j + 1} u)_t\\
 &= (-1)^r\left[\sum_{\ell=1}^\lambda (\Delta^r v)_{\ell + r}u_{t +\ell} - \sum_{\ell = 1}^\lambda(\Delta^r v)_{\ell + r + 1}u_{t + \ell} - (\Delta^r v)_{r + 1} u _t\right]\\
 &\qquad + \sum_{j = 1}^r (-1)^j (\Delta^{j - 1} v)_j (\Delta^{r - j + 1} u)_t\\
 &= (-1)^{r+1}\sum_{\ell = 1}^\lambda (\Delta^{r+1} v)_{\ell + r + 1}u_{t + \ell} + (-1)^{r+1}(\Delta^r v)_{r+1} u_t + \sum_{j = 1}^r(-1)^j(\Delta^{j-1} v)_j (\Delta^{r + 1 - j} u)_t.
\end{align*}
In the fourth equality, we use the fact that $(\Delta^{r} v)_{\lambda + r + 1} = \sum_{j = 0}^r (-1)^j \binom{r}{j}v_{\lambda + r + 1 - j} = 0$. 
\end{proof}

We are now ready to prove Lemma \ref{lem:bound VD}. Let $V_j$ be the $j$th row of the matrix $V$,
then by Lemma \ref{summation by parts}
\begin{align*}
	V_jD^r u &= \sum_{\ell = 1}^\lambda v_\ell (\Delta^r u)_{(j - 1)\lambda + \ell}\\
	&= \underbrace{(-1)^r \sum_{\ell = 1}^\lambda (\Delta^r v)_{\ell + r}u_{t + \ell}}_{(I)} + \underbrace{\sum_{j = 1}^r (-1)^j (\Delta^{j - 1} v)_j (\Delta^{r - j} u)_t}_{(II)},
\end{align*}
where $t = (j - 1)\lambda.$

To bound $(I)$ and $(II)$, we first define the Fourier series of a sequence $z := (z_j)_{j\in\Z}$ in $\ell^2(\Z)$ as
\[
	\widehat{z}(\xi) := \sum_{j\in\Z} z_j e^{-i j \xi}.
\]
Let $h := (h_j)_{j\in \Z}$ where $h_j = (-1)^j \binom{r}{j}$, for $j = 0, \ldots, r$, and $0$ otherwise. We observe that
\begin{align*}
|(I)| &\leq \|u\|_\infty \sum_{\ell = 1}^\lambda |(\Delta^r v)_{\ell + r}|\\
&= \|u\|_\infty \sum_{\ell = 1}^\lambda |(h*v)_{\ell + r}|\\
&\leq \|u\|_\infty \|h*v\|_1,
\end{align*}
where $(h*v)_\ell = \sum_{j\in\Z}h_j v_{\ell -j}$ denotes the convolution of $h$ and $v$. 
Since
\[
	\widehat{h}(\xi) = \sum_{j\in\Z}h_j e^{-ij\xi} = \sum_{j = 0}^r (-1)^j \binom{r}{j}(e^{-i\xi})^j = (1 - e^{-i\xi})^r
\]
and
\[
	\widehat{v}(\xi) = \sum_{j\in\Z}v_j e^{-ij\xi} = (1 + e^{-i\xi} + \ldots + e^{-i\xi(\tilde{\lambda} - 1)})^r,
\]
we obtain
\[
	\widehat{h*v}(\xi) = \widehat{h}(\xi)\widehat{v}(\xi) = (1 - e^{-i\xi\tilde{\lambda}})^r = \sum_{j=0}^r(-1)^j\binom{r}{j}e^{-i\xi\tilde{\lambda} j},
\]
which implies that the non-zero entries of $h*v$ are $(h*v)_{\tilde{\lambda} j} = (-1)^j \binom{r}{j}$. Thus, $\|h*v\|_1 = 2^r$ 
and \[|(I)| \leq 2^r\|u\|_\infty.\]

To bound $(II)$, we observe that $|(\Delta^{j-1}v)_j|\leq 2^{j-1}\max_{1\leq j\leq r}|v_j|$ 
and $|(\Delta^{r-j}u)_t|\leq 2^{r-j} \|u\|_\infty$, thus
\[
	|(II)| \leq r2^{r-1}\max_{1\leq j\leq r}|v_j| \|u\|_\infty \leq r2^{r-1}2^{2r-1}\|u\|_\infty,
\]
where the last inequality is due to the fact that by the definition of $v$, $v_j$ are multinomial coefficients satisfying 
$\max_{1\leq j \leq r} |v_j| \leq \binom{2r-1}{r-1} = \frac{1}{2}\binom{2r}{r} \leq \frac{1}{2} 4^r$.
Combining the bounds on  $(I)$ and $(II)$, we get
\[
	\frac{|V_j D^r u|}{\|u\|_\infty} \leq  2^r + r2^{3r - 2} \leq r2^{3r - 1}.
\]
This yields
\begin{equation}\label{bound VDr}
\|\widetilde{V}D^r \|_{\infty \to 2} \leq \frac{9}{8}\frac{r2^{3r - 1}}{\|v\|_2}\leq (8r)^{r+1}\lambda^{-r + 1/2}.
\end{equation}

Here, the first and second inequalities in \eqref{bound VDr} are due to the facts that 
\[\|\widetilde{V}D^r\|_{\infty \to 2} \leq \sqrt{p}\|\widetilde{V}D^r\|_{\infty \to \infty} = \frac{9}{8\|v\|_2}\|VD^r\|_{\infty \to \infty}\] 
and $\|v\|_2 \geq \lambda^{r - 1/2}r^{-r}$, respectively.

\section*{Acknowledgments}
RS was supported in part by the NSF via DMS-1517204. The authors would like to thank Sjoerd Dirksen and Laurent Jacques for stimulating conversations on binary embeddings. 
\bibliographystyle{abbrv}
\bibliography{refs}

\begin{thebibliography}{10}

\bibitem{AC2009}
N.~Ailon and B.~Chazelle.
\newblock The fast johnson--lindenstrauss transform and approximate nearest
  neighbors.
\newblock {\em SIAM Journal on computing}, 39(1):302--322, 2009.

\bibitem{AL2013}
N.~Ailon and E.~Liberty.
\newblock An almost optimal unrestricted fast johnson-lindenstrauss transform.
\newblock {\em ACM Transactions on Algorithms (TALG)}, 9(3):21, 2013.

\bibitem{BDD2008}
R.~Baraniuk, M.~Davenport, R.~DeVore, and M.~Wakin.
\newblock A simple proof of the restricted isometry property for random
  matrices.
\newblock {\em Constr. Approx.}, 28(3):253--263, 2008.

\bibitem{BFNPW2017}
R.~G. Baraniuk, S.~Foucart, D.~Needell, Y.~Plan, and M.~Wootters.
\newblock Exponential decay of reconstruction error from binary measurements of
  sparse signals.
\newblock {\em IEEE Transactions on Information Theory}, 63(6):3368--3385,
  2017.

\bibitem{baraniuk2017exponential}
R.~G. Baraniuk, S.~Foucart, D.~Needell, Y.~Plan, and M.~Wootters.
\newblock Exponential decay of reconstruction error from binary measurements of
  sparse signals.
\newblock {\em IEEE Transactions on Information Theory}, 63(6):3368--3385,
  2017.

\bibitem{BPY2006}
J.~J. Benedetto, A.~M. Powell, and O.~Yilmaz.
\newblock Sigma-delta quantization and finite frames.
\newblock {\em IEEE Transactions on Information Theory}, 52(5):1990--2005,
  2006.

\bibitem{BLPY2010}
J.~Blum, M.~Lammers, A.~M. Powell, and O.~Y{\i}lmaz.
\newblock Sobolev duals in frame theory and sigma-delta quantization.
\newblock {\em J. Fourier Anal. Appl.}, 16(3):365--381, 2010.

\bibitem{BB2007}
P.~Boufounos and R.~Baraniuk.
\newblock Quantization of sparse representations.
\newblock In {\em Data Compression Conference, 2007. DCC '07}, pages 378--378.
  March 2007.

\bibitem{BJKS2015}
P.~T. Boufounos, L.~Jacques, F.~Krahmer, and R.~Saab.
\newblock Quantization and compressive sensing.
\newblock In {\em Compressed sensing and its applications}, Appl. Numer.
  Harmon. Anal., pages 193--237. Birkh\"auser/Springer, Cham, 2015.

\bibitem{B2014}
J.~Bourgain.
\newblock An improved estimate in the restricted isometry problem.
\newblock In {\em Geometric aspects of functional analysis}, volume 2116 of
  {\em Lecture Notes in Math.}, pages 65--70. Springer, Cham, 2014.

\bibitem{BGBL2008}
M.~Bridson, T.~Gowers, J.~Barrow-Green, and I.~Leader.
\newblock Geometric and combinatorial group theory.
\newblock {\em The Princeton companion to mathematics, T. Gowers, J.
  Barrow-Green, and I. Leader, Eds}, page~10, 2008.

\bibitem{CTR2006}
E.~J. Cand\`es, J.~K. Romberg, and T.~Tao.
\newblock Stable signal recovery from incomplete and inaccurate measurements.
\newblock {\em Comm. Pure Appl. Math.}, 59(8):1207--1223, 2006.

\bibitem{CT2006}
E.~J. Candes and T.~Tao.
\newblock Near-optimal signal recovery from random projections: universal
  encoding strategies?
\newblock {\em IEEE Trans. Inform. Theory}, 52(12):5406--5425, 2006.

\bibitem{carl1985}
B.~Carl.
\newblock Inequalities of bernstein-jackson-type and the degree of compactness
  of operators in banach spaces.
\newblock {\em Ann. Inst. Fourier (Grenoble)}, 35(3):79--118, 1985.

\bibitem{CGV2013}
M.~Cheraghchi, V.~Guruswami, and A.~Velingker.
\newblock Restricted isometry of {F}ourier matrices and list decodability of
  random linear codes.
\newblock {\em SIAM J. Comput.}, 42(5):1888--1914, 2013.

\bibitem{CCBJKL2016}
A.~Choromanska, K.~Choromanski, M.~Bojarski, T.~Jebara, S.~Kumar, and Y.~LeCun.
\newblock Binary embeddings with structured hashed projections.
\newblock In {\em International Conference on Machine Learning}, pages
  344--353, 2016.

\bibitem{Chou2013}
E.~Chou.
\newblock {\em Beta-{D}uals of {F}rames and {A}pplications to {P}roblems in
  {Q}uantization}.
\newblock ProQuest LLC, Ann Arbor, MI, 2013.
\newblock Thesis (Ph.D.)--New York University.

\bibitem{CG2016}
E.~Chou and C.~S. G\"unt\"urk.
\newblock Distributed noise-shaping quantization: {I}. {B}eta duals of finite
  frames and near-optimal quantization of random measurements.
\newblock {\em Constr. Approx.}, 44(1):1--22, 2016.

\bibitem{CG2017}
E.~Chou and C.~S. G{\"u}nt{\"u}rk.
\newblock {\em Distributed Noise-Shaping Quantization: II. Classical Frames},
  pages 179--198.
\newblock Springer International Publishing, Cham, 2017.

\bibitem{CGKRO2015}
E.~Chou, C.~S. G\"unt\"urk, F.~Krahmer, R.~Saab, and O.~Y{\i}lmaz.
\newblock Noise-shaping quantization methods for frame-based and compressive
  sampling systems.
\newblock In {\em Sampling theory, a renaissance}, Appl. Numer. Harmon. Anal.,
  pages 157--184. Birkh\"auser/Springer, Cham, 2015.

\bibitem{DD2003}
I.~Daubechies and R.~DeVore.
\newblock Approximating a bandlimited function using very coarsely quantized
  data: a family of stable sigma-delta modulators of arbitrary order.
\newblock {\em Ann. of Math. (2)}, 158(2):679--710, 2003.

\bibitem{DGK2011}
P.~Deift, F.~Krahmer, and C.~S. Güntürk.
\newblock An optimal family of exponentially accurate one-bit sigma-delta
  quantization schemes.
\newblock {\em Communications on Pure and Applied Mathematics}, 64(7):883--919,
  2011.

\bibitem{DJR2017}
S.~Dirksen, H.~C. Jung, and H.~Rauhut.
\newblock One-bit compressed sensing with partial gaussian circulant matrices.
\newblock {\em arXiv preprint arXiv:1710.03287}, 2017.

\bibitem{DS2016}
S.~Dirksen and A.~Stollenwerk.
\newblock Fast binary embeddings with gaussian circulant matrices: improved
  bounds.
\newblock {\em Discrete Comput. Geom., to appear, arXiv:1608.06498}, 2016.

\bibitem{Donoho2006}
D.~L. Donoho.
\newblock Compressed sensing.
\newblock {\em IEEE Trans. Inform. Theory}, 52(4):1289--1306, 2006.

\bibitem{FK2014}
J.~Feng and F.~Krahmer.
\newblock An rip approach to sigma-delta quantization for compressed sensing.
\newblock {\em IEEE Signal Process. Lett.}, 21(11):1351--1355, 2014.

\bibitem{FKS2017}
J.-M. Feng, F.~Krahmer, and R.~Saab.
\newblock Quantized compressed sensing for partial random circulant matrices.
\newblock {\em arXiv preprint arXiv:1702.04711}, 2017.

\bibitem{FR2013}
S.~Foucart and H.~Rauhut.
\newblock {\em A mathematical introduction to compressive sensing}.
\newblock Applied and Numerical Harmonic Analysis. Birkh\"auser/Springer, New
  York, 2013.

\bibitem{GL2013}
Y.~Gong, S.~Lazebnik, A.~Gordo, and F.~Perronnin.
\newblock Iterative quantization: A procrustean approach to learning binary
  codes for large-scale image retrieval.
\newblock {\em IEEE Transactions on Pattern Analysis and Machine Intelligence},
  35(12):2916--2929, 2013.

\bibitem{GVT1998}
V.~Goyal, M.~Vetterli, and N.~Thao.
\newblock Quantized overcomplete expansions in ir/sup n: analysis, synthesis,
  and algorithms.
\newblock {\em IEEE Transactions on Information Theory}, 44(1):16--31, 1998.

\bibitem{Gunturk2003}
C.~S. G\"unt\"urk.
\newblock One-bit sigma-delta quantization with exponential accuracy.
\newblock {\em Comm. Pure Appl. Math.}, 56(11):1608--1630, 2003.

\bibitem{GLPSY2013}
C.~S. G\"unt\"urk, M.~Lammers, A.~M. Powell, R.~Saab, and O.~Y{\i}lmaz.
\newblock Sobolev duals for random frames and {$\Sigma\Delta$} quantization of
  compressed sensing measurements.
\newblock {\em Found. Comput. Math.}, 13(1):1--36, 2013.

\bibitem{HHL2011}
J.~P. Haldar, D.~Hernando, and Z.-P. Liang.
\newblock Compressed-sensing mri with random encoding.
\newblock {\em IEEE transactions on Medical Imaging}, 30(4):893--903, 2011.

\bibitem{HBRN2010}
J.~Haupt, W.~U. Bajwa, G.~Raz, and R.~Nowak.
\newblock Toeplitz compressed sensing matrices with applications to sparse
  channel estimation.
\newblock {\em IEEE Trans. Inform. Theory}, 56(11):5862--5875, 2010.

\bibitem{HR2016}
I.~Haviv and O.~Regev.
\newblock The restricted isometry property of subsampled fourier matrices.
\newblock In {\em Geometric Aspects of Functional Analysis}, pages 163--179.
  Springer, 2017.

\bibitem{H2016}
T.~Huynh.
\newblock {\em Accurate Quantization in Redundant Systems: From Frames to
  Compressive Sampling and Phase Retrieval}.
\newblock PhD thesis, New York University, 2016.

\bibitem{Jacques2015}
L.~Jacques.
\newblock Small width, low distortions: quantized random embeddings of
  low-complexity sets.
\newblock {\em IEEE Transactions on Information Theory}, 63(9):5477--5495,
  2017.

\bibitem{JC2013}
L.~Jacques and V.~Cambareri.
\newblock Time for dithering: fast and quantized random embeddings via the
  restricted isometry property.
\newblock {\em Information and Inference: A Journal of the IMA}, 2017.

\bibitem{JLBB2013}
L.~Jacques, J.~N. Laska, P.~T. Boufounos, and R.~G. Baraniuk.
\newblock Robust 1-bit compressive sensing via binary stable embeddings of
  sparse vectors.
\newblock {\em IEEE Trans. Inform. Theory}, 59(4):2082--2102, 2013.

\bibitem{JL}
W.~B. Johnson and J.~Lindenstrauss.
\newblock Extensions of {L}ipschitz mappings into a {H}ilbert space.
\newblock In {\em Conference in modern analysis and probability ({N}ew {H}aven,
  {C}onn., 1982)}, volume~26 of {\em Contemp. Math.}, pages 189--206. Amer.
  Math. Soc., Providence, RI, 1984.

\bibitem{KMR2014}
F.~Krahmer, S.~Mendelson, and H.~Rauhut.
\newblock Suprema of chaos processes and the restricted isometry property.
\newblock {\em Comm. Pure Appl. Math.}, 67(11):1877--1904, 2014.

\bibitem{KSW2012}
F.~Krahmer, R.~Saab, and R.~Ward.
\newblock Root-exponential accuracy for coarse quantization of finite frame
  expansions.
\newblock {\em IEEE Trans. Inform. Theory}, 58(2):1069--1079, 2012.

\bibitem{KSY2014}
F.~Krahmer, R.~Saab, and O.~Y{\i}lmaz.
\newblock Sigma-{D}elta quantization of sub-{G}aussian frame expansions and its
  application to compressed sensing.
\newblock {\em Inf. Inference}, 3(1):40--58, 2014.

\bibitem{KW2011}
F.~Krahmer and R.~Ward.
\newblock New and improved johnson--lindenstrauss embeddings via the restricted
  isometry property.
\newblock {\em SIAM Journal on Mathematical Analysis}, 43(3):1269--1281, 2011.

\bibitem{LPY2010}
M.~Lammers, A.~M. Powell, and O.~Y{\i}lmaz.
\newblock Alternative dual frames for digital-to-analog conversion in
  sigma-delta quantization.
\newblock {\em Adv. Comput. Math.}, 32(1):73--102, 2010.

\bibitem{LWKC2011}
W.~Liu, J.~Wang, S.~Kumar, and S.-F. Chang.
\newblock Hashing with graphs.
\newblock In {\em Proceedings of the 28th international conference on machine
  learning (ICML-11)}, pages 1--8, 2011.

\bibitem{LDP2007}
M.~Lustig, D.~Donoho, and J.~M. Pauly.
\newblock Sparse mri: The application of compressed sensing for rapid mr
  imaging.
\newblock {\em Magnetic resonance in medicine}, 58(6):1182--1195, 2007.

\bibitem{MPT2008}
S.~Mendelson, A.~Pajor, and N.~Tomczak-Jaegermann.
\newblock Uniform uncertainty principle for {B}ernoulli and subgaussian
  ensembles.
\newblock {\em Constr. Approx.}, 28(3):277--289, 2008.

\bibitem{MRW2016}
S.~Mendelson, H.~Rauhut, and R.~Ward.
\newblock Improved bounds for sparse recovery from subsampled random
  convolutions.
\newblock {\em arXiv preprint arXiv:1610.04983}, 2016.

\bibitem{Murphy2012}
M.~Murphy, M.~Alley, J.~Demmel, K.~Keutzer, S.~Vasanawala, and M.~Lustig.
\newblock Fast $\ell\_1 $-spirit compressed sensing parallel imaging mri:
  Scalable parallel implementation and clinically feasible runtime.
\newblock {\em IEEE transactions on medical imaging}, 31(6):1250--1262, 2012.

\bibitem{NPW2014}
J.~Nelson, E.~Price, and M.~Wootters.
\newblock New constructions of {RIP} matrices with fast multiplication and
  fewer rows.
\newblock In {\em Proceedings of the {T}wenty-{F}ifth {A}nnual {ACM}-{SIAM}
  {S}ymposium on {D}iscrete {A}lgorithms}, pages 1515--1528. ACM, New York,
  2014.

\bibitem{ORS2015}
S.~Oymak, B.~Recht, and M.~Soltanolkotabi.
\newblock Isometric sketching of any set via the restricted isometry property.
\newblock {\em Information and Inference, to appear}, 2018.

\bibitem{OTH2017}
S.~Oymak, C.~Thrampoulidis, and B.~Hassibi.
\newblock Near-optimal sample complexity bounds for circulant binary embedding.
\newblock In {\em Acoustics, Speech and Signal Processing (ICASSP), 2017 IEEE
  International Conference on}, pages 6359--6363. IEEE, 2017.

\bibitem{PRT2013}
G.~E. Pfander, H.~Rauhut, and J.~A. Tropp.
\newblock The restricted isometry property for time-frequency structured random
  matrices.
\newblock {\em Probab. Theory Related Fields}, 156(3-4):707--737, 2013.

\bibitem{PV2013}
Y.~Plan and R.~Vershynin.
\newblock Robust 1-bit compressed sensing and sparse logistic regression: a
  convex programming approach.
\newblock {\em IEEE Trans. Inform. Theory}, 59(1):482--494, 2013.

\bibitem{PV2014}
Y.~Plan and R.~Vershynin.
\newblock Dimension reduction by random hyperplane tessellations.
\newblock {\em Discrete Comput. Geom.}, 51(2):438--461, 2014.

\bibitem{RL2009}
M.~Raginsky and S.~Lazebnik.
\newblock Locality-sensitive binary codes from shift-invariant kernels.
\newblock In {\em Advances in neural information processing systems}, pages
  1509--1517, 2009.

\bibitem{Rauhut2008}
H.~Rauhut.
\newblock Stability results for random sampling of sparse trigonometric
  polynomials.
\newblock {\em IEEE Transactions on Information Theory}, 54(12):5661--5670, Dec
  2008.

\bibitem{RRT2012}
H.~Rauhut, J.~Romberg, and J.~A. Tropp.
\newblock Restricted isometries for partial random circulant matrices.
\newblock {\em Applied and Computational Harmonic Analysis}, 32(2):242--254,
  2012.

\bibitem{Romberg2009}
J.~Romberg.
\newblock Compressive sensing by random convolution.
\newblock {\em SIAM Journal on Imaging Sciences}, 2(4):1098--1128, 2009.

\bibitem{RV2008}
M.~Rudelson and R.~Vershynin.
\newblock On sparse reconstruction from {F}ourier and {G}aussian measurements.
\newblock {\em Comm. Pure Appl. Math.}, 61(8):1025--1045, 2008.

\bibitem{SWY2017_2}
R.~Saab, R.~Wang, and {\"O}.~Y{\i}lmaz.
\newblock From compressed sensing to compressed bit-streams: practical
  encoders, tractable decoders.
\newblock {\em IEEE Transactions on Information Theory}, 2017.

\bibitem{SWY2018}
R.~Saab, R.~Wang, and O.~Y{\i}lmaz.
\newblock Quantization of compressive samples with stable and robust recovery.
\newblock {\em Applied and Computational Harmonic Analysis}, 44(1):123 -- 143,
  2018.

\bibitem{SH2009}
R.~Salakhutdinov and G.~Hinton.
\newblock Semantic hashing.
\newblock {\em International Journal of Approximate Reasoning}, 50(7):969--978,
  2009.

\bibitem{ST2005}
R.~Schreier and G.~C. Temes.
\newblock {\em {Understanding delta-sigma data converters}}.
\newblock Wiley, New York, NY, 2005.

\bibitem{talagrand2006generic}
M.~Talagrand.
\newblock {\em The generic chaining: upper and lower bounds of stochastic
  processes}.
\newblock Springer Science \& Business Media, 2006.

\bibitem{TWDBB2006}
J.~A. Tropp, M.~B. Wakin, M.~F. Duarte, D.~Baron, and R.~G. Baraniuk.
\newblock Random filters for compressive sampling and reconstruction.
\newblock In {\em Acoustics, Speech and Signal Processing, 2006. ICASSP 2006
  Proceedings. 2006 IEEE International Conference on}, volume~3, pages
  III--III. IEEE, 2006.

\bibitem{Vasanawala2010}
S.~S. Vasanawala, M.~T. Alley, B.~A. Hargreaves, R.~A. Barth, J.~M. Pauly, and
  M.~Lustig.
\newblock Improved pediatric mr imaging with compressed sensing.
\newblock {\em Radiology}, 256(2):607--616, 2010.

\bibitem{W2016}
R.~Wang.
\newblock Sigma delta quantization with harmonic frames and partial fourier
  ensembles.
\newblock {\em Journal of Fourier Analysis and Applications}, Dec 2017.

\bibitem{WTF2009}
Y.~Weiss, A.~Torralba, and R.~Fergus.
\newblock Spectral hashing.
\newblock In {\em Advances in neural information processing systems}, pages
  1753--1760, 2009.

\bibitem{YCP2015}
X.~Yi, C.~Caramanis, and E.~Price.
\newblock Binary embedding: Fundamental limits and fast algorithm.
\newblock In F.~Bach and D.~Blei, editors, {\em Proceedings of the 32nd
  International Conference on Machine Learning}, volume~37 of {\em Proceedings
  of Machine Learning Research}, pages 2162--2170, Lille, France, 07--09 Jul
  2015. PMLR.

\bibitem{YBKGC2015}
F.~X. Yu, A.~Bhaskara, S.~Kumar, Y.~Gong, and S.-F. Chang.
\newblock On binary embedding using circulant matrices.
\newblock {\em arXiv preprint arXiv:1511.06480}, 2015.

\end{thebibliography}

\end{document}